\documentclass[american,aps,onecolumn,reprint,superscriptaddress]{revtex4-1}
\usepackage[T1]{fontenc}
\usepackage[latin9]{inputenc}
\usepackage{amsmath}
\usepackage{amsthm}
\usepackage{amssymb}
\usepackage{graphicx}
\usepackage{float}
\usepackage{epstopdf}

\usepackage[normalem]{ulem}
\makeatletter

\newtheoremstyle{customStyle1}  
{0pt}       
{0pt}       
{\normalfont}   
{\parindent}        
{\em}  
{. --}   	 
{.5em}       
{\thmname{#1}\thmnumber{ #2}\thmnote{ (#3)}}  

\theoremstyle{plain}
\newtheorem{thm}{\protect\theoremname}
\newtheorem{prop}[thm]{\protect\propositionname}
\newtheorem{lem}[thm]{\protect\lemmaname}
\newtheorem{cor}[thm]{\protect\corollaryname}
\newtheorem{thmMain}{\protect\theoremname}

\newtheorem*{thm*}{\protect\theoremname}
\newtheorem*{prop*}{\protect\propositionname}
\newtheorem*{lem*}{\protect\lemmaname}
\newtheorem*{cor*}{\protect\corollaryname}

\ifx\proof\undefined
\newenvironment{proof}[1][\protect\proofname]{\par
	\normalfont\topsep6\p@\@plus6\p@\relax
	\trivlist
	\itemindent\parindent
	\item[\hskip\labelsep\scshape #1]\ignorespaces
}{%
	\endtrivlist\@endpefalse
}
\providecommand{\proofname}{Proof}
\fi

\usepackage{babel}
\usepackage{times}
\usepackage{pxfonts}
\usepackage{color}

\definecolor{mygrey}{gray}{0.35}
\definecolor{myblue}{rgb}{0.2,0.2,0.8}
\definecolor{myzard}{cmyk}{0,0,0.05,0}
\definecolor{mywhite}{rgb}{1,1,1}
\definecolor{myred}{rgb}{0.9,0.1,0.}
\definecolor{kangdaorange}{rgb}{1,0.4,0.}
\usepackage[colorlinks=true,citecolor=myblue,linkcolor=myblue,urlcolor=myblue]{hyperref}

\newcommand{\tr}{\operatorname{\bf{tr}}} 
\newcommand{\Real}{\operatorname{\bf{Re}}} 
\newcommand{\Imag}{\operatorname{\bf{Im}}} 

\renewcommand{\vec}[1]{\text{\boldmath$#1$}}

\newcommand{\bra}[1]{\langle #1|}
\newcommand{\ket}[1]{|#1\rangle}
\newcommand{\braket}[2]{\langle #1|#2\rangle}
\newcommand{\ketbra}[2]{|#1\rangle\!\langle#2|}

\makeatother

\usepackage{babel}
\providecommand{\corollaryname}{Corollary}
\providecommand{\lemmaname}{Lemma}
\providecommand{\propositionname}{Proposition}
\providecommand{\theoremname}{Theorem}

\begin{document}

\title{Quantum coherence and state conversion: theory and experiment}

\author{Kang-Da Wu}

\thanks{These authors contributed equally to this work.}
\affiliation{CAS Key Laboratory of Quantum Information, University of Science and Technology of China, \\ Hefei 230026, People's Republic of China}
\affiliation{CAS Center For Excellence in Quantum Information and Quantum Physics, University of Science and Technology of China, Hefei, 230026, People's Republic of China}

\author{Thomas Theurer}

\thanks{These authors contributed equally to this work.}
\affiliation{Institute of Theoretical Physics, Universität Ulm, Albert-Einstein-Allee 11, D-89069 Ulm, Germany}

\author{Guo-Yong Xiang}

\email{gyxiang@ustc.edu.cn}
\affiliation{CAS Key Laboratory of Quantum Information, University of Science and Technology of China, \\ Hefei 230026, People's Republic of China}
\affiliation{CAS Center For Excellence in Quantum Information and Quantum Physics, University of Science and Technology of China, Hefei, 230026, People's Republic of China}

\author{Chuan-Feng Li}

\affiliation{CAS Key Laboratory of Quantum Information, University of Science and Technology of China, \\ Hefei 230026, People's Republic of China}
\affiliation{CAS Center For Excellence in Quantum Information and Quantum Physics, University of Science and Technology of China, Hefei, 230026, People's Republic of China}

\author{Guang-Can Guo}

\affiliation{CAS Key Laboratory of Quantum Information, University of Science and Technology of China, \\ Hefei 230026, People's Republic of China}
\affiliation{CAS Center For Excellence in Quantum Information and Quantum Physics, University of Science and Technology of China, Hefei, 230026, People's Republic of China}

\author{Martin B. Plenio}

\affiliation{Institute of Theoretical Physics, Universität Ulm, Albert-Einstein-Allee 11, D-89069 Ulm, Germany}

\author{Alexander Streltsov}

\email{a.streltsov@cent.uw.edu.pl}
\affiliation{Centre for Quantum Optical Technologies IRAU, Centre of New Technologies, University of Warsaw, Banacha 2c, 02-097 Warsaw, Poland}
\affiliation{Faculty of Applied Physics and Mathematics, Gda\'{n}sk University of Technology, 80-233 Gda\'{n}sk, Poland}
\affiliation{National Quantum Information Centre in Gda\'{n}sk, 81-824 Sopot, Poland}

\maketitle
\textbf{The resource theory of coherence~\cite{BaumgratzPhysRevLett.113.140401,Levi_2014,WinterPhysRevLett.116.120404,YadinPhysRevX.6.041028,StreltsovRevModPhys.89.041003,aberg2006quantifying} studies the operational value of superpositions in quantum technologies. A key question in this theory concerns the efficiency of manipulation and inter-conversion~\cite{WinterPhysRevLett.116.120404,ChitambarPhysRevLett.117.030401,ChitambarPhysRevLett.116.070402,StreltsovPhysRevX.7.011024,PhysRevLett.119.140402}
of the resource. Here we solve this question completely for qubit states by determining the optimal probabilities for mixed state conversions via stochastic incoherent operations.
Extending the discussion to distributed scenarios, we introduce and address the task of \emph{assisted incoherent state conversion} where the process is enhanced by making use of correlations with a second party. Building on these results, we demonstrate experimentally that the optimal state conversion probabilities can be achieved in a linear optics set-up. This paves the way towards real world applications of coherence transformations in current quantum technologies.
}

Practical constraints on our ability to manipulate physical systems restrict the control we can exert on them. It is, for example, exceedingly difficult to exchange quantum systems undisturbed over long distances~\cite{HorodeckiRevModPhys.81.865}. When manipulating spatially separated subsystems, effectively, this limits us to Local Operations and Classical Communication (LOCC).  Under these operations, it is only possible to prepare certain states, i.e. separable ones. The states which cannot be produced under LOCC are called entangled and elevated to resources: Consuming them allows to implement operations such as quantum state teleportation~\cite{PhysRevLett.70.1895} to achieve perfect quantum state transfer which would not be possible with LOCC alone. This has important consequences, e.g. in quantum communication and other quantum technologies, but also for our understanding of the view of the fundamental laws of nature~\cite{EinsteinPhysRev.47.777,PlenioPhysRevLett.78.2275,plenio2007introduction,HorodeckiRevModPhys.81.865}. 

Entanglement is explored within the framework of quantum resource theories, which can also be used to investigate other non-classical  features of quantum mechanics in a systematic way. A concept underlying many facets of non-classicality, including entanglement, is the superposition principle. Since a quantum system naturally decoheres in the presence of unavoidable interactions between the system and its environment~\cite{ZurekRevModPhys.75.715,BromleyPhysRevLett.114.210401}, superposition is itself a resource, which is studied in the recently developed resource theory of quantum coherence~\cite{BaumgratzPhysRevLett.113.140401,WinterPhysRevLett.116.120404,YadinPhysRevX.6.041028,StreltsovRevModPhys.89.041003,aberg2006quantifying}. In this framework, the set of free operations analogous to LOCC in entanglement theory are incoherent operations (IOs), corresponding to quantum measurements which cannot create coherence even if postselection is applied on the individual measurement outcomes~\cite{BaumgratzPhysRevLett.113.140401}.
	
One of the central questions in any resource theory is the state conversion problem, i.e. the characterization of the possible interconversion of resources under transformations allowed by the corresponding resource theory. The answer to this question leads to a preorder on the resource states which determines their usefulness or value in technological applications, since a given state can be used in all protocols which require a state that can be created from it. This opens new perspectives on how such resources can lead to practical advantages in quantum metrology~\cite{PhysRevLett.116.150502,MarvianPhysRevA.94.052324}, quantum algorithms~\cite{HilleryPhysRevA.93.012111,Matera2058-9565-1-1-01LT01}, and even quantum dynamics in biology~\cite{doi:10.1080/00405000.2013.829687}. 

In this work, we provide a full solution to the coherence conversion problem for qubit systems, the fundamental building blocks in quantum information: We determine the optimal conversion probability between two states via IOs. When the technology to build quantum computers becomes available, it is likely that they will appear initially in small numbers. They have complete control over their qubits and can assist a less powerful remote client restricted to IOs. In particular they can assist him in state conversions, which we study under the name of assisted incoherent state conversion, solving the problem of optimality for two-qubit pure and Werner states. Moreover, we demonstrate the first experimental realization of non-unitary IOs using photonic quantum technologies and show their capability of implementing optimal state conversion on qubits both with and without assistance. This is an important step towards the experimental investigation and systematic manipulation of coherence in quantum technological applications.

\section*{Results}

\subsection* {Theoretical framework}
In this section, we describe the foundations of this work.
In the resource theory of coherence, an orthonormal basis of states $\{\ket{i}\}$, usually motivated on physical grounds as being easy to synthesize or store, are considered classical. Any mixture of such states, $\rho=\sum_{i}p_{i}\ket{i}\!\bra{i}$, is referred to as "free" and termed \emph{incoherent}, similar to probability distributions on classical states. The free operations are referred to as \emph{incoherent operations} (IO)~\cite{BaumgratzPhysRevLett.113.140401}: these are quantum
transformations $\Lambda$ which admit an incoherent Kraus decomposition
$\Lambda[\rho]=\sum_{i}K_{i}\rho K_{i}^{\dagger}$ with incoherent Kraus operators $K_{i}$; i.e., $K_{i}\ket{m}\sim\ket{n}$
for incoherent states $\ket{m}$ and $\ket{n}$. To implement a stochastic IO, we formally postselect a deterministic IO according to the measurement outcomes~$i$. Obviously, in the presence of such restrictions, the amount of coherence in a quantum system can never increase.

Most of the analysis presented in the following can be reduced to the mathematically simpler family of \emph{strictly incoherent operations} (SIO). These are operations that can be decomposed into strictly incoherent Kraus operators $K_i$, which are defined by the property that both $K_{i}$ and $K_{i}^{\dagger}$ are incoherent, and correspond to quantum processes which do not use coherence~\cite{WinterPhysRevLett.116.120404,YadinPhysRevX.6.041028}. First, we will solve the problem of state conversion for qubits under the restricted sets of operations IO and SIO theoretically. Then, we extend our analysis to the problem of \emph{assisted incoherent state conversion}, which we introduce now as a game between two parties, Alice and Bob. Initially, they share a bipartite quantum state $\rho^{AB}$, and the aim of the game is to establish a certain state $\sigma^B$ on Bob's side. Clearly, if all quantum transformations were allowed locally, Bob could achieve this task by simply erasing his local system and preparing the desired state $\sigma^B$. However, the situation changes if Bob is constrained to local IOs: in this case he cannot prepare the state $\sigma^B$ if the state has coherence. Moreover, as we will show later, correlations in the joint state $\rho^{AB}$ can be used to enhance Bob's conversion possibilities, if Alice assists Bob by measuring her particle and communicating the measurement outcome.\\

\begin{figure*}[htp]
	\label{fig:exp}
	\includegraphics[scale=0.042]{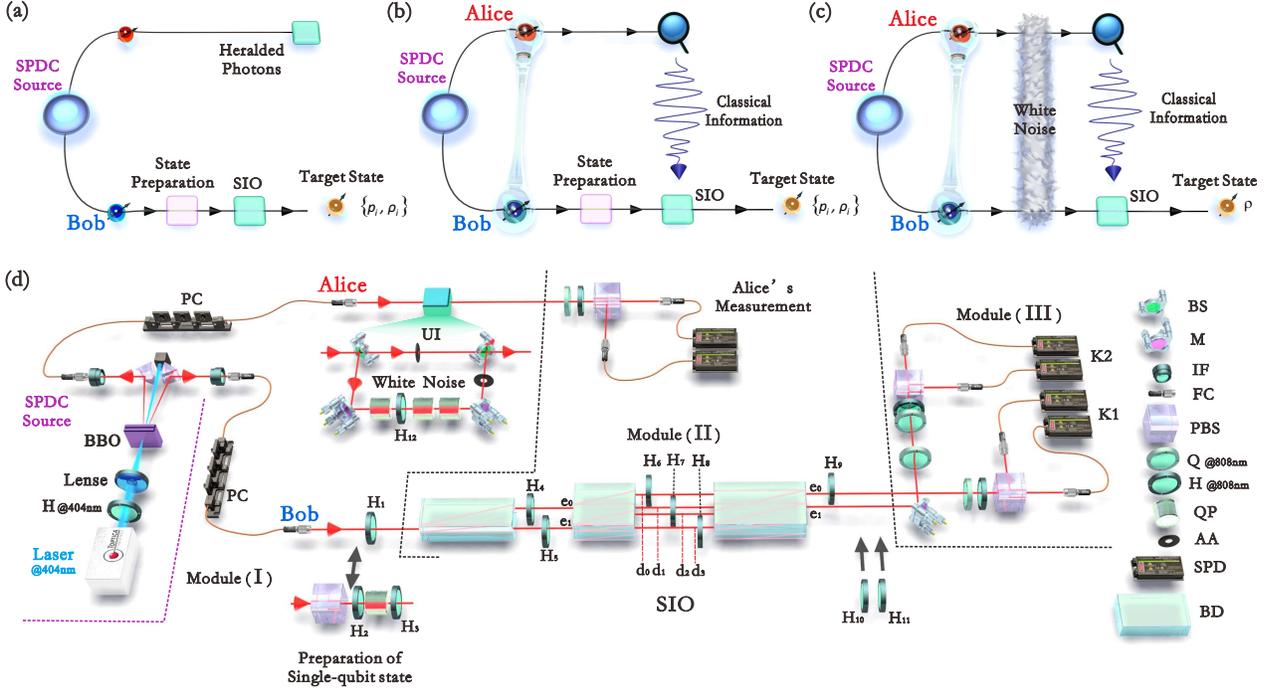}
	\caption{\label{fig:exp} \textbf{Experimental protocols and setup}. In (a-c), the three experiments performed in our laboratory are illustrated: (a) single-qubit conversion without assistance; (b) pure entangled two-qubit state conversion with assistance; (c) noisy two-qubit state conversion with assistance. The whole setup in (d) can be divided into three modules: (I) state preparation; (II) incoherent state conversion with or without assistance; (III) tomography. In (I), we can prepare a class of single-qubit states as in Eq.~(\ref{qubitmixstate}) for Bob, its purification shared with Alice, and a class of Werner states; in (II), we experimentally implement the incoherent operations, both with and without assistance; in (III), we identify the quantum states of Bob. The optical components appearing in the setup are $\beta-$barium borate (BBO), half-wave plate (H$_i$), quarter-wave plate (Q), beam displacer (BD), adjustable aperture (AA), interference filter (IF), beam splitter (BS), mirror (M), quartz plate (QP), polarizing beam splitter (PBS), polarization controller (PC), single photon detector (SPD), fiber coupler (FC), and unbalanced interferometer (UI). $K_{1,2}$ denotes the outcomes of Kraus operators 1(2).}
\end{figure*}

\subsection*{Optimal conversion without assistance}
For a general qubit state $\rho$, the exact shape of the state space which can be achieved by IOs is described in the following Theorem, making use of the Bloch sphere representation which we introduce in the Methods Section, where $\vec{r}$ and $\vec{s}$ are the Bloch vectors of the initial and the final state, respectively, and using the definition $r^2=r_x^2+r_y^2$. 
\begin{thm}\label{thm:Main}
A qubit state $\sigma$ is reachable via a stochastic SIO or IO transformation from a fixed initial qubit state $\rho$ with a given probability $p$ iff
\begin{subequations}
		\begin{align}
			&r^2 s_{z}^{2}+\left(1-r_{z}^{2}\right)s^2\le r^2, \label{eq:firstThmMain}\\
			&p^2s^2 \le \frac{r^2}{1+|r_z|}\left[2p-(1-|r_z|)\right]. \label{eq:secondThmMain}
		\end{align}
	\end{subequations}

\end{thm}

\noindent For a geometrical interpretation of this Theorem and further discussion we refer to the Methods Section. These results allow us to evaluate the optimal conversion probability $P\left(\rho\rightarrow\sigma\right)$ via IO and SIO for any two states $\rho$ and $\sigma$ of a single qubit. It holds that $P\left(\rho\rightarrow\sigma\right)=0$ if $r^2 s_{z}^{2}+\left(1-r_{z}^{2}\right)s^2>r^2$ and otherwise
\begin{align}
	P\left(\rho\rightarrow\sigma\right) = \min\left\{ \frac{r^2}{\left(1+|r_z|\right)s^2}\left[1+\sqrt{1-\frac{s^2\left(1-r_z^2\right)}{r^2}}\right],1\right\} \label{eq:Probability}
	\end{align}
Moreover, these theoretical results can be extended beyond qubits, leading to a necessary condition for state conversion via stochastic SIO, see Supplementary Information for more details.

\subsection*{Optimal conversion with assistance}
We now present our results with assistance, for pure entangled states and a class of mixed states. The task of assisted incoherent state conversion is equivalent to transforming a shared state $\rho^{AB}$ into a local state $\sigma^B$ on Bob's side via \emph{local quantum-incoherent operations and classical communication } (LQICC)~\cite{ChitambarPhysRevLett.116.070402,StreltsovPhysRevX.7.011024}.
These operations consist of general local operations on Alice's side, local IOs on Bob's side and the exchange of measurement results via a classical channel. The problem of optimal conversion of a general two qubit entangled state $\ket{\psi}^{AB}$ into an arbitrary local state $\sigma^B$ is solved in the following Theorem.

\begin{thm}
	\label{thm:PurificationAssistance}
	Let Alice and Bob share a pure two-qubit state  $\ket{\psi}^{AB}$ and denote Bob's local state by $\rho^B$. The maximal probability  $P_\mathrm{a}(\ket{\psi}^{AB}\rightarrow\sigma^{B})$ to prepare the qubit state $\sigma^B$ on Bob's side via one-way LQICC is given by
	
	\begin{equation}
	P_\mathrm{a}\left(\ket{\psi}^{AB}\rightarrow\sigma^{B}\right)=
	\min\left\{ 1,\left(1-\left|r_z\right|\right)\frac{1+\sqrt{1-s^2}}{s^2}\right\},\label{eq:PurificationAssistance}
	\end{equation}
	where $\vec{r}$ and $\vec{s}$ are the Bloch vectors of $\rho^B$ and $\sigma^B$, respectively.
\end{thm}	

When the state is subjected to noise, the probabilities of assisted incoherent state conversions are reduced. As an example, we consider the two-qubit Werner state,
\begin{equation}\label{wernerstate}
\rho_{\mathrm{w}}^{AB}=q_\mathrm{w}\ket{\phi^{+}}\!\bra{\phi^{+}}+(1-q_\mathrm{w})\frac{\openone}{4},
\end{equation}
with the maximally entangled state $\ket{\phi^{+}}=(\ket{00}+\ket{11})/\sqrt{2}$.
The optimal conversion probability for converting $\rho_{\mathrm{w}}^{AB}$ into the qubit state $\sigma^B$ via one-way LQICC can also be evaluated explicitly by,
\begin{equation} \label{eq:Werner-2}
P_\mathrm{a}\left(\rho_{\mathrm{w}}^{AB}\rightarrow\sigma^{B}\right)=\begin{cases}
1 & \mathrm{if}\,\,\,q_\mathrm{w}\geq \frac{s^2}{\sqrt{1-s_z^2}},\\  
0 & \mathrm{otherwise},
\end{cases}
\end{equation}
where $\vec{s}$ denotes the Bloch vector of $\sigma^B$.

Generally, correlations in the joint two-qubit state $\rho^{AB}$ always enhance the conversion possibilities of Bob whenever the state is not quantum-incoherent, i.e., not of the form $\rho^{AB}=\sum_{i}p_{i}\rho_{i}^{A}\otimes\ket{i}\!\bra{i}^{B}$, see Supplementary Information for more details.\\

\subsection*{Experimental results} 

\begin{figure*}[htp]
	\includegraphics[scale=0.13]{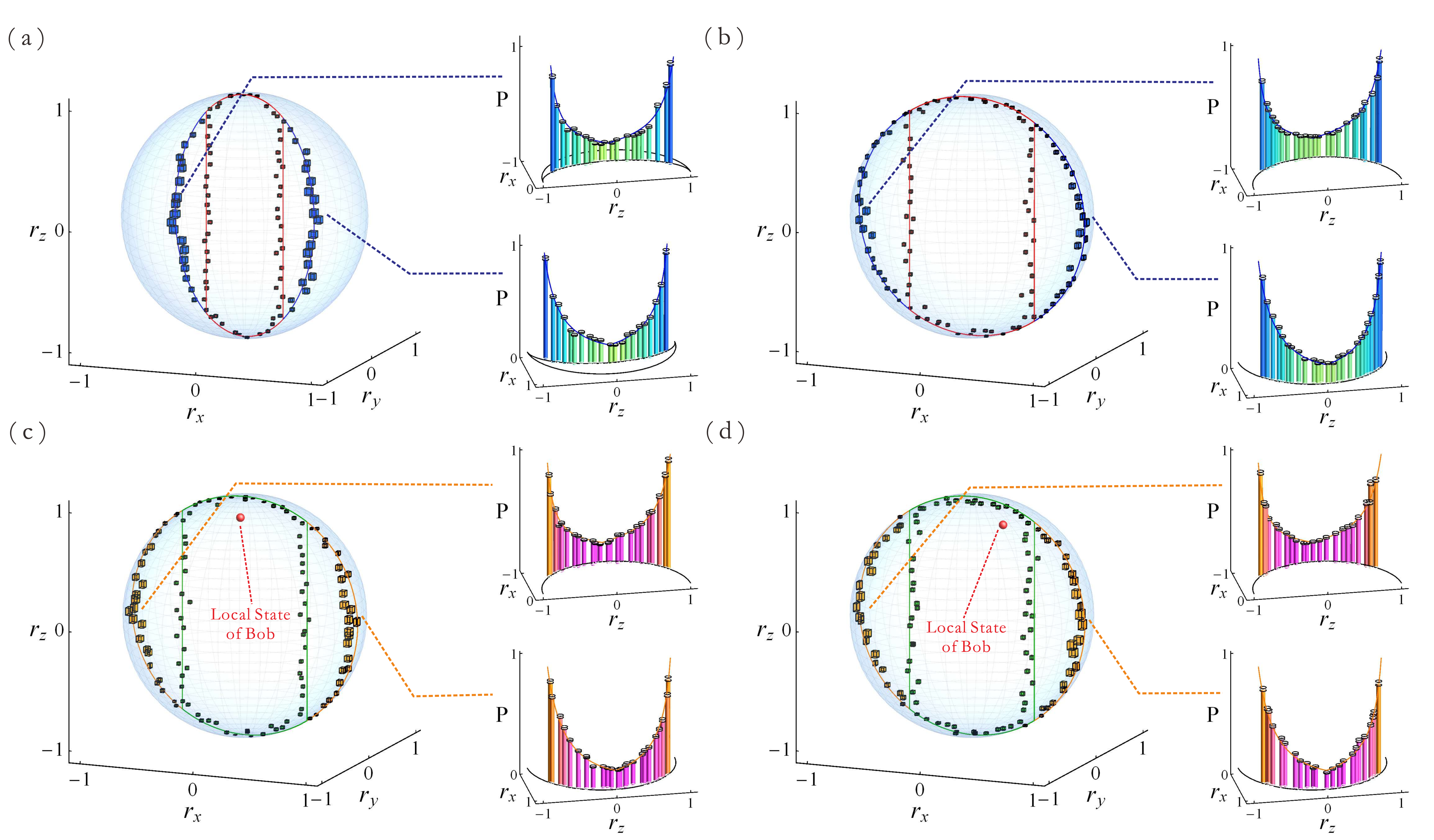}
	\caption{\label{fig:data2} \textbf{Experimental results for state conversion: single-qubit states and pure two-qubit states}. Experimental results for two local states for Bob as in Eq.~(\ref{qubitmixstate}) with Bloch coordinates $\left(\frac{1}{3},0,\frac{5}{6}\right)$ and $\left(\frac{\sqrt{11}}{6},0,\frac{5}{6}\right)$, are shown in (a,b) without assistance from Alice. The states are prepared with high fidelity up to 0.999. In the left of (a,b), the DC and SC boundaries in $x$-$z$ plane are shown in red and blue cubes, respectively, with each side representing the variance $\delta\left\langle r_i \right\rangle (i=x,y,z)$ derived from Poisson distribution of single photons. In the right of (a,b), conversion probabilities $P(\rho\!\rightarrow\!\sigma)$ for boundaries of SC are shown with respect to the $x$-$z$ Bloch coordinates of the target state $\sigma^B$. Experimental results for two local states for Bob, with Bloch coordinates $\left(0,0,\frac{5}{6}\right)$ and $\left(\frac{1}{3},0,\frac{5}{6}\right)$, sharing pure entangled states $\ket{\psi}^{AB}$ with Alice, are shown in (c,d). The experimental two-qubit states are prepared with fidelity of 0.989 and 0.982. In the left of (c,d), the DC and SC boundaries in $x$-$z$ plane are shown in green and orange cubes, respectively. In the right of (c,d), the probabilities of conversion $P_a(\ket{\psi}^{AB}\rightarrow\ket{\phi}^B)$ are shown with respect to the $x$-$z$ Bloch coordinates of the target state $\ket{\phi}^B$. Solid lines represent theoretical predictions.
	}
\end{figure*}

We experimentally implement the above protocols on several classes of input states, both with and without assistance. The experimental protocols and setup are illustrated in Fig.~\ref{fig:exp}. The setup consists of three modules, detailed descriptions can be found in the Methods and Supplementary Information.

For verifying the theoretical predictions of incoherent single-qubit state conversion, we experimentally initialize Bob's photon as
\begin{equation}\label{qubitmixstate}
\rho^B=\frac{1}{2}\left(\openone + r_x\sigma_x+r_z\sigma_z \right),
\end{equation}
where $r_x=\frac{1}{3}, r_z=\frac{5}{6}$ for a mixed input and $r_x=\frac{\sqrt{11}}{6}, r_z=\frac{5}{6}$ for a pure input. Our goal is now to convert the initial states using the IOs available in our experiments.
In Fig.~\ref{fig:data2}(a, b), the experimentally obtained boundary of state space for deterministic conversion (DC) is shown for the $x$-$z$ plane by red cubes. The boundaries with respect to stochastic conversion (SC) are shown by blue cubes.
All solid lines represent theoretical predictions from Theorem~\ref{thm:Main}. Also the experimental conversion probability is plotted as a function of $x$-$z$ coordinates of the target state $\sigma$ (blue cylinders) for two states. The solid lines represent the theoretical predictions from Eq.~(\ref{eq:Probability}).
There is a fundamental difference between pure and mixed inputs: A pure and coherent input can be converted via IOs into any qubit state at least stochastically. For mixed input, this is not the case.

To explore assisted conversions experimentally, we first let Bob share a pure entangled state $\ket{\psi}^{AB}$ with Alice, where Bob's local state is the state in Eq.~(\ref{qubitmixstate}) with Bloch coordinates $\left(0,0,\frac{5}{6}\right)$ and $\left(\frac{1}{3},0,\frac{5}{6}\right)$. Our experimental results are shown in Fig.~\ref{fig:data2}(c, d).
Remarkably, due to Theorem~\ref{thm:PurificationAssistance} the probability for assisted conversion only depends on the $z$-coordinate of the initial state, which explains the close similarity of both the state spaces and conversion probabilities in Fig.~\ref{fig:data2}(c, d). We also experimentally test two Werner states, one with entanglement and one without.
In Fig.~\ref{fig:data4}, the DC boundaries are shown in purple and blue for the two states. In accordance with Eq.~(\ref{eq:Werner-2}), resorting to SCs does not allow to prepare additional states.

\begin{figure}[htp]
	\label{fig:data4}
	\includegraphics[scale=0.12]{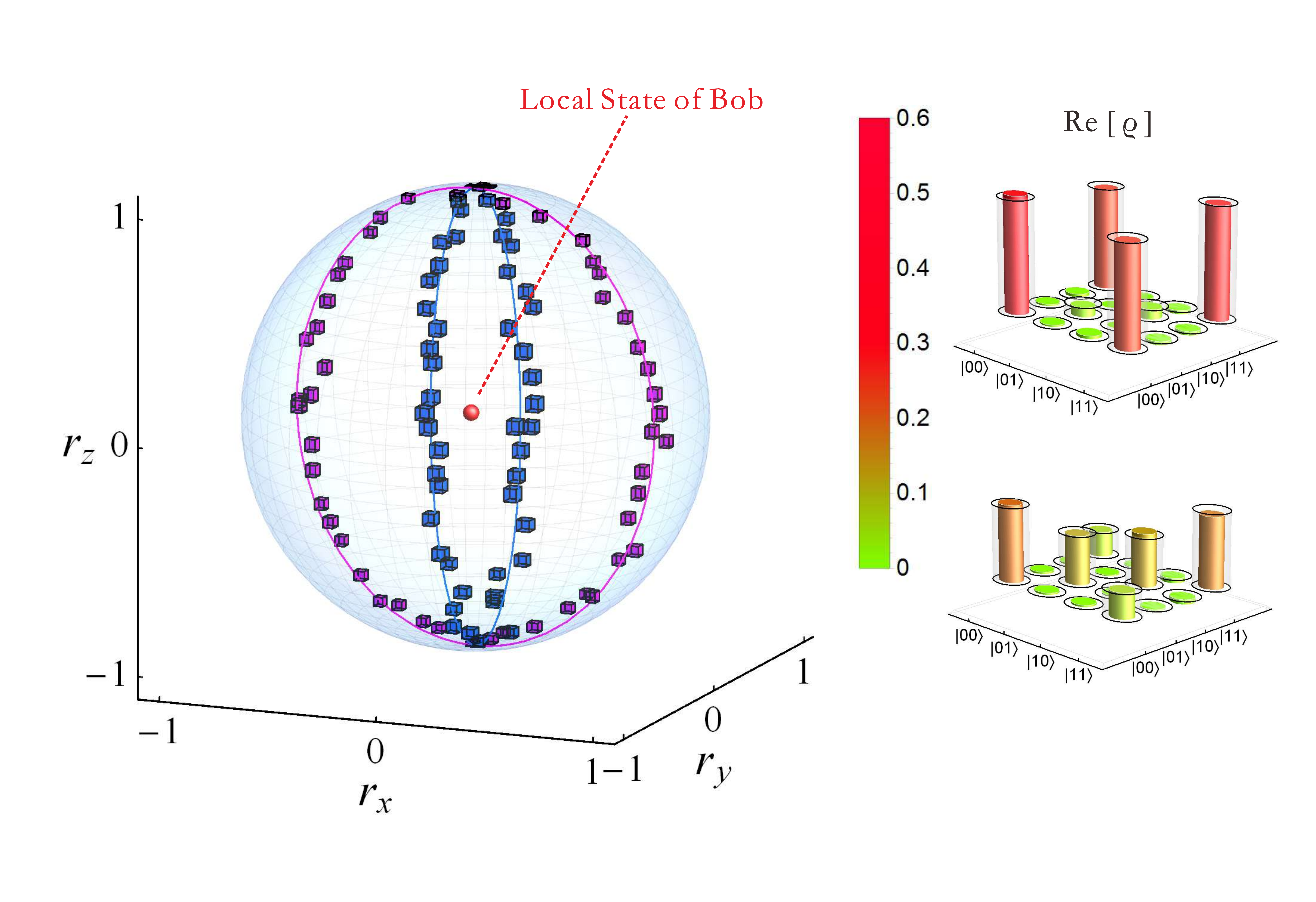}
	\caption{\label{fig:data4} \textbf{Experimental results for assisted incoherent state conversion: noisy two-qubit states}. Experimental results for local states for Bob, $\rho^B=\frac{1}{2}\openone$, sharing pure entangled states $\ket{\psi}^{AB}$ with Alice, subjected to a controllable proportion of white noise. Though Bob will find his system in a maximally mixed state with zero coherence, he can prepare certain coherent states if he takes advantage of Alice's assistance. In our experiment, we use two Werner states with $q_\mathrm{w}=0.8245$ and $0.2075$. In the right, the real parts of the tomographically reconstructed quantum states $\rho^{AB}_\mathrm{w}$ are shown,  with a fidelity of $0.986$ and $0.997$, respectively. In the left, the DC boundaries in $x$-$z$ plane for conversions $\rho_\mathrm{w}^{AB}\rightarrow\sigma^B$ are shown in purple and blue respectively.	
	}
\end{figure}

Compared to DC, Bob can expand the state space achieved with SC, and even obtain all qubit states by taking advantage of assistance. The expansion can be seen from the results in Fig.~\ref{fig:data1}, for the single-qubit state in Eq.~(\ref{qubitmixstate}) of Bob, and its purification $\ket{\psi}^{AB}$. In Fig.~\ref{fig:data1}(a), we experimentally show the boundary of accessible state space, both deterministically and stochastically, with and without assistance. Noting that the $\ell_1$ norm coherence~\cite{BaumgratzPhysRevLett.113.140401} reads: $C_{\ell_1}(\rho)=\sum_{i\neq j}|\rho_{ij}|=r$ for qubit states, we can obtain a similar relation between achievable coherence and maximal conversion probability as in Eq.~(\ref{eq:Probability}, \ref{eq:PurificationAssistance}), the experimental results are shown in 
Fig.~\ref{fig:data1}(b). Although local coherence can never be increased deterministically, we can still exceed the original coherence at the expense of success probability. A maximally coherent state $\ket{+}$ can be obtained by utilizing SC and assistance.

\section*{Discussion}
In this work we study the problem of quantum state conversion within the resource theory of quantum coherence, both theoretically and experimentally. The state conversion problem is important in any resource theory, since it determines the value of states for protocols using the resource under study. The result presented here are a significant generalization of recent results on single-shot coherence theory~\cite{PhysRevLett.119.140402,Regula1711.10512,Vijayan1804.06554,PhysRevLett.121.070404,SciRepCoherenceTransQubit} and single-shot resource theories in general~\cite{GourPhysRevA.95.062314}, and include necessary conditions on the existence of stochastic conversions, which we generalized to higher dimensions.

In most resource theories one is also interested in the possibilities of asymptotic state conversion, where many instances of the initial and final state are available. As we show in detail in the Supplementary Information, our results also pave the way towards a complete solution of this problem: our single-shot conversion rate gives a lower bound on
the asymptotic conversion rate, which is in some areas significantly
better than the best previously known bound~\cite{WinterPhysRevLett.116.120404}. In addition, it coincides for some states with an upper bound from \cite{WinterPhysRevLett.116.120404}, solving the asymptotic conversion problem in these cases. Moreover, the results allow us to investigate the irreversibly of coherence theory in the asymptotic limit and to determine the possible distillable coherence for fixed coherence cost.

Experimentally implementing non-unitary incoherent operations for the first time, we demonstrated that a quantum optical experiment can closely achieve the expected optimal conversion rates. The corresponding optical setup should be seen as a building block for more general transformations, also going beyond single qubits and IOs. The results presented in this work can then serve as benchmarks for these more advanced setups.

\begin{figure*}[htp]
	\label{fig:data1}
	\includegraphics[scale=0.163]{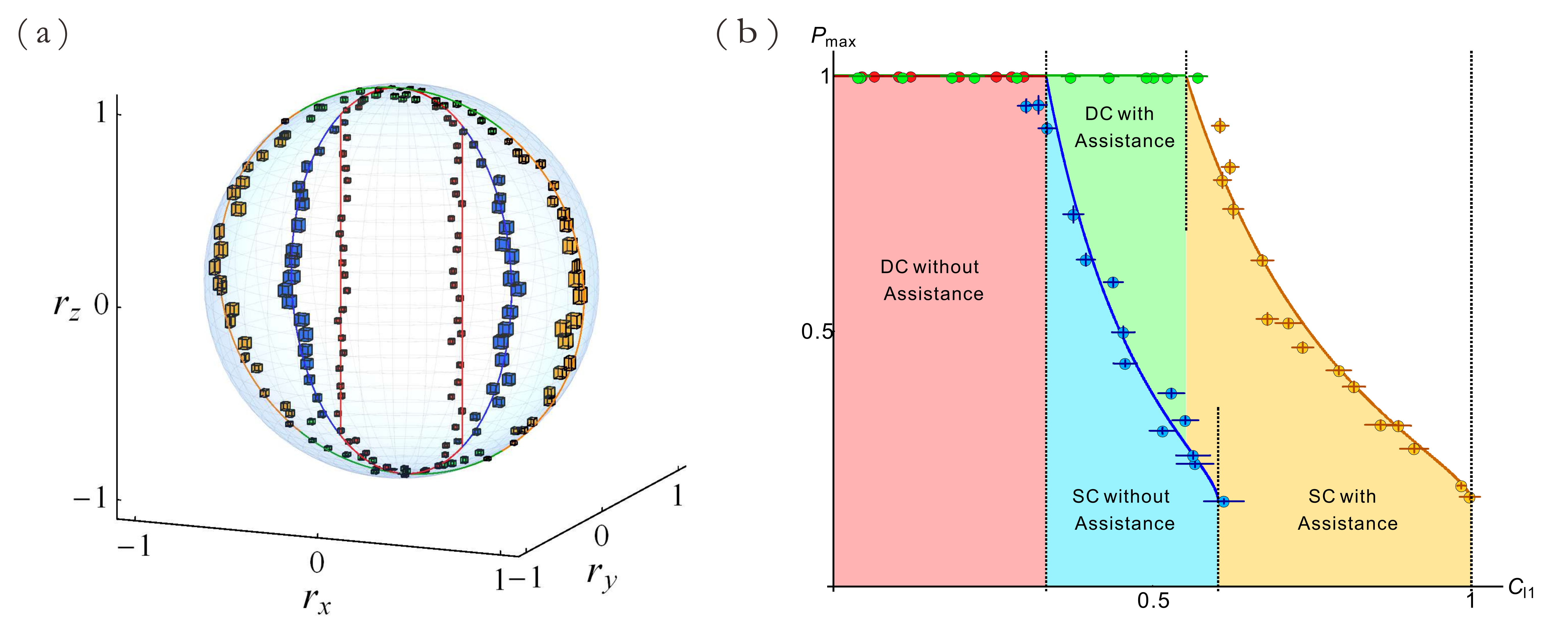}
	\caption{\label{fig:data1} \textbf{Experimental results for showing the capability of enlarging conversion boundaries via different protocols}. The local state of Bob is experimentally prepared as $\rho^B=\frac{1}{2}\left(\openone + \frac{1}{3}\sigma_x+\frac{5}{6}\sigma_z \right)$. In (a), we can see that the accessible states of Bob can be enlarged by using different conversion protocols, the red boundary can be achieved via DC without assistance, which shows the basic capability of local incoherent conversion. When we use SC, without assistance, we can make the conversion boundary larger, shown as blue. Combining the boundary of SC and DC, we obtain an ellipsoid in the Bloch space. With assistance from a pure source, we can enlarge the conversion boundary to the surface of the Bloch sphere. The boundaries of assisted conversion, both DC and SC, are shown as green and orange respectively. As the boundaries are rotationally invariant with respect to $z$, we focus on the $x$-$z$ plane by taking a round cross-section. In (b), the maximal success probabilities versus obtained $\ell_1$ norm coherence are plotted for these different protocols.}	
\end{figure*}

\section*{Methods}
{\bf{Bloch representation.}} Most of this work is concerned with qubits, which allows us to make frequent use of the Bloch representation stating that every qubit state $\rho$ can be represented by a subnormalized vector $\vec{r}=(r_x,r_y,r_z)$ through
$\rho=(\openone + \vec{r} \cdot \vec{\sigma})/2$,
where $\vec{\sigma}$ represents a vector containing the Pauli matrices.  As done above, we denote density operators by small Greek letters and their Bloch vectors by the respective small Latin letter.
Throughout the paper, we assume the eigenbasis of $\sigma_z$ to be incoherent. Then rotations about the $z$-axis of the Bloch sphere and their inverse are both in SIO and in IO, leading to an invariance of  conversion probabilities under these rotations. This makes it very convenient to introduce the quantity
\begin{equation}
r=\sqrt{r_x^2+r_y^2},
\end{equation}
corresponding to the distance of the state to the incoherent axis.\\

{\bf{Geometrical interpretation of Theorem ~\ref{thm:Main}.}} Theorem~\ref{thm:Main} has a convenient geometrical interpretation on the Bloch sphere: for fixed $\rho$, Eq.~(\ref{eq:firstThmMain}) defines an ellipsoid which is \emph{independent} of $p$ and Eq.~(\ref{eq:secondThmMain}) a cylinder which \emph{depends} on $p$. The states to which $\rho$ can be converted with probability $p$ lie inside their intersection. For $p \le 1-|r_z|$, the ellipsoid is entirely contained in the cylinder and Eq.~(\ref{eq:secondThmMain}) is automatically satisfied if Eq.~(\ref{eq:firstThmMain}) holds (see proof of Thm.~\ref{thm:Main}). Therefore, lowering the demanded probability of success below $1-|r_z|$ will not increase the set of reachable states. This implies that for mixed $\rho$, there is a discontinuity in the optimal conversion probability $P(\rho\!\rightarrow\!\sigma)$ and the states outside the ellipsoid cannot be achieved via stochastic IOs, even with arbitrary little probability. This also implies that the states outside the ellipsoid cannot be approximated, because no state in a neighborhood can be reached.\\

{\bf{Experimental setup.}} In Module (I) of (d), we experimentally prepare an entangled state 
\begin{equation}
\ket{\psi(\theta)}^{AB}=\cos2\theta\ket{00}+\sin2\theta\ket{11}
\end{equation}
via a \textit{spontaneous parametric down conversion} (SPDC) process, with arbitrary $\theta$, determined by the angle of the 404~nm H. The states 0 and 1 are encoded in the polarization degree. One of the two photons is sent to Alice for multiple uses; the other is sent to Bob for state preparation. 
In the case of (a), we implement state conversion without assistance, and $\theta$ is set to $0^\circ$. The second photon is used to initialize Bob's qubit in the desired state as in Eq.~(\ref{qubitmixstate}) by the combination of H$_{2,3}$, a PBS, and a QP. In this case, the first photon (Alice) is used as a trigger. 
In the case of (b), the first photon is sent to Alice, which will allow her to assist Bob. Two-qubit entangled states 
\begin{equation}
\ket{\psi}^{AB}=\sqrt{\mu_0}\ket{0}^A\ket{\beta_0}^B+\sqrt{\mu_1}\ket{1}^A\ket{\beta_1}^B
\end{equation}
can be prepared with Bob's local state being described in Eq.~(\ref{qubitmixstate}), where $\mu_{0, 1}$ and $\ket{\beta_{0, 1}}^B$ denote eigenvalues and eigenvectors of Bob's local state. 
In the case of (c), we can experimentally prepare mixed two qubit states consisting of a pure Bell state $\ket{\phi^+}^{AB}$ and a controllable amount of white noise $\frac{1}{4}\openone$. This is achieved by using unbalanced interference in Alice's arm. 
In Module (II) of (d), a class of SIOs are implemented on Bob's photons, by the combination of 6 Hs and 3 BDs (details in the Supplementary Information). In the case of (b,c), these operations can depend on the result of measurements made on Alice's qubit~\cite{wu2017experimentally,PRLresourceonversion}.
In Module (III) of (d), Bob can perform quantum state tomography~\cite{Qi13quantum} to identify the target states $\{p_i, \rho_i\}$ in his hand. A detailed description of the experimental setup can be found in the Supplementary Information.

We acknowledge useful discussions with Dario Egloff, Swapan Rana, Christine Silberhorn, and Jan Sperling. K.-D.W., G.-Y.X., C.-F.L., and G.-C.G. acknowledge support from the National Nature Science Foundation of China (NSFC; 11574291 and 11774334), National Key R\&D Program (2016YFA0301700), and Anhui Initiative in Quantum Information Technologies. MBP and TT acknowledge support by the ERC Synergy Grant BioQ (grant no 319130) and the BMBF project Q.Link.X.
AS acknowledges financial support by the ''Quantum Optical Technologies'' project, carried out within the International Research Agendas programme of the Foundation for Polish Science co-financed by the European Union under the European Regional Development Fund, the National
Science Center in Poland (POLONEZ UMO-2016/21/P/ST2/04054), and the
European Union's Horizon 2020 research and innovation programme
under the Marie Sk\l{}odowska-Curie grant agreement No. 665778.

\bibliography{bibliography}

\clearpage
\appendix

\section*{Supplementary Information}
In this Supplementary Information, we provide an extended discussion of both the theory and the experiments presented in the main text. In the first section, we gather useful theoretical results, which we apply in the second section to prove our theoretical findings stated in the main text, partially extending them beyond qubits. The third section is dedicated to an extended discussion of the implications of these findings, focusing on asymptotic conversions and the irreversibly of coherence theory. The fourth section describes the experimental aspect in more detail, which include state preparation, implementation of the incoherent operations and tomography of the output states.

\section{Useful results}
As mentioned in the main text, to implement a stochastic incoherent operation, we formally postselect a deterministic incoherent operation according to the measurement outcomes $i$. Now assume we deal with a stochastic operation that can be decomposed into incoherent Kraus operators which are not necessarily complete, i.e., $\sum_{i}K_{i}^{\dagger}K_{i}\leq\openone$, and transforms a state $\rho$ into the state
$\rho\rightarrow\Lambda[\rho]/p$ with conversion probability $p=\mathrm{Tr}(\Lambda[\rho])$. If we want to call this operation incoherent, we have to ensure that it is part of a deterministic incoherent operation, otherwise we would simply disregard the nonfree part of the operation. 
That this is always possible has been shown in \cite{PhysRevLett.119.230401}. Therefore we call all stochastic operations that can be decomposed into incoherent Kraus operators incoherent as well. If we can implement a stochastic transformations from a state $\rho$ to a state $\sigma$ with probability $p$, we will write $\rho\rightarrow p\sigma$. 

As we will see in the following and announced in the main text, most of the analysis in this work
can be reduced to the mathematically simpler family of \emph{strictly incoherent operations}. As in the case of IO, a free completion is possible, which is the content of the following Proposition.

\begin{prop}\label{prop:compSIO} Every stochastic quantum operation that can be decomposed into strictly incoherent Kraus operators is part of a deterministic SIO.
\end{prop}
\begin{proof}
	Strictly incoherent Kraus operators $K_n$ are of the form
	\begin{align}
	K_{n}=\sum_{i}c_{i,n}\ketbra{f_{n}(i)}{i}
	\end{align}
	where $f_{n}(i)$ is a bijective function on $\{1,...,d\}$. If they form a stochastic quantum operation, we have
	\begin{align}
	\sum_{n}K_{n}^{\dagger}K_{n}=\sum_{n,i}|c_{i,n}|^{2}\ketbra{i}{i}\le\openone.
	\end{align}
	Therefore $\sum_{n}|c_{i,n}|^{2}\le1\forall i$ and we can define
	\begin{align}
	\tilde{c_{i}}=\sqrt{1-\sum_{n}|c_{i,n}|^{2}}
	\end{align}
	and 
	\begin{align}
	\tilde{K}=\sum_{i}\tilde{c_{i}}\ketbra{i}{i},
	\end{align}
	which is a strictly incoherent Kraus operator and has the property 
	\begin{align}
	\tilde{K}^{\dagger}\tilde{K}+\sum_{n}K_{n}^{\dagger}K_{n}=\openone.
	\end{align}
\end{proof}

\noindent This allows to prove the following:
\begin{prop} \label{prop:MixingSIO}
	For two states $\rho$,
	$\sigma$ and a probability $p$ let there be a stochastic SIO achieving the transformation 
	\begin{equation}
	\rho\rightarrow p\sigma.
	\end{equation}
	Then, for every incoherent state $\tau$ and every $0\leq q\leq1-p$,
	there exists a stochastic SIO achieving the transformation 
	\begin{equation}
	\rho\rightarrow p\sigma+q\tau.
	\end{equation}
\end{prop}
\begin{proof}
	The key idea in this proof is that the set of strictly incoherent Kraus operators is closed under concatenation. Therefore, the overall map that describes the application of a SIO on post-selected output states of another SIO is still in SIO. From Prop.~\ref{prop:compSIO} follows that we can always complete a stochastic SIO for free. The part completing the map has, with probability $1-p$, a state $\mu$ as an output. Applying total dephasing to $\mu$, we obtain an incoherent state $\mu'$, which we can transform into $\tau$ using SIO. In addition, we can do this only stochastically, which proves the Proposition. 
\end{proof} 
Now we are ready to show the equivalence of IO and SIO with respect to probabilistic qubit state transformations.
\begin{thm}\label{thm:equivSIOandIO}
	Let $\rho$ and $\sigma$ be states of a single qubit. The following
	statements are equivalent:\\
	(1) There exists an IO implementing $\rho\rightarrow p\sigma$.\\
	(2) There exists a SIO implementing $\rho\rightarrow p\sigma$.
\end{thm}

\begin{proof}
	An incoherent Kraus operator $K$ is of the form
	\begin{align}
	K=\sum_i c(i) \ketbra{j(i)}{i},
	\end{align}
	and it is strictly incoherent if $j(i)$ is one-to-one~\cite{WinterPhysRevLett.116.120404}.
	Therefore all incoherent qubit Kraus operators are either also strictly incoherent or their output is, independent of the input, incoherent. Let us use the strictly incoherent ones to define a stochastic SIO. Then Prop.~\ref{prop:MixingSIO} finishes the proof. Note that this proof technique does not work in higher dimensions, since then, there exist $j(i)$ that have neither the same output for all $i$ (and have thus incoherent output), nor are they one-to-one.
\end{proof}	

Another result we will need later is stated in the following Lemma.
\begin{lem}
	\label{lem:2}For any correlated state $\rho^{AB}\neq\rho^{A}\otimes\rho^{B}$
	there exists a two-element POVM $\{M^{A},\openone^{A}-M^{A}\}$ on
	Alice's side such that $\mathrm{Tr}\left[M^{A}\rho^{AB}\right]>0$
	and 
	\begin{equation}
	\frac{\mathrm{Tr}_{A}\left[M^{A}\rho^{AB}\right]}{\mathrm{Tr}\left[M^{A}\rho^{AB}\right]}\neq\rho^{B}:=\mathrm{Tr}_A\left[\rho^{AB}\right].\label{eq:Lemma2}
	\end{equation}
\end{lem}
\begin{proof}
	In the following, let $\{\ket{i}^{B}\}$ be an eigenbasis of Bob's
	state $\rho^{B}$. Consider now states of the form 
	\begin{equation}
	\rho^{AB}=\sum_{i}p_{i}\rho_{i}^{A}\otimes\ket{i}\!\bra{i}^{B}.\label{eq:QI}
	\end{equation}
	Following the notion of~\cite{ChitambarPhysRevLett.116.070402},
	we call these states quantum-incoherent with respect to an eigenbasis
	of Bob. 
	
	If the state $\rho^{AB}$ shared by Alice and Bob is not of the form~(\ref{eq:QI}),
	there exists a POVM element $M^{A}$ such that $\mathrm{Tr}\left[M^{A}\rho^{AB}\right]>0$
	and $\mathrm{Tr}_{A}\left[M^{A}\rho^{AB}\right]$ is not diagonal
	in the basis $\{\ket{i}^{B}\}$, see Thm.~2 in \cite{ChitambarPhysRevLett.116.070402}.
	This proves the Lemma for states $\rho^{AB}$ which are not quantum-incoherent
	with respect to an eigenbasis of Bob.
	
	For completing the proof, let now $\rho^{AB}$ be of the form~(\ref{eq:QI}).
	If the state $\rho^{AB}$ is not a product state, the decomposition~(\ref{eq:QI})
	must have at least two different states $\rho_{i}^{A}\neq\rho_{j}^{A}$.
	Because these states are different, there necessarily exists a POVM
	element $M^{A}$ such that 
	\begin{equation}
	\mathrm{Tr}\left[M^{A}\rho_{i}^{A}\right]\neq\mathrm{Tr}\left[M^{A}\rho_{j}^{A}\right].\label{eq:Inequality}
	\end{equation}
	In the last step of the proof, note that Eq.~(\ref{eq:Inequality})
	implies the following:
	\begin{equation}	\frac{\braket{i}{\rho^{B}|i}}{\braket{j}{\rho^{B}|j}}\neq\frac{\braket{i}{\mathrm{Tr}_{A}\left[M^{A}\rho^{AB}\right]|i}}{\braket{j}{\mathrm{Tr}_{A}\left[M^{A}\rho^{AB}\right]|j}}.
	\end{equation}
	Thus, also in this case Eq.~(\ref{eq:Lemma2}) is fulfilled, and
	the proof of the Lemma is complete.
\end{proof}

\section{Technical proofs}

Here we give the missing proofs of the results in the main text and their partial generalizations beyond qubits. For readability, we restate the results.

\begin{thmMain}
	A qubit state $\sigma$ is reachable via a stochastic SIO or IO transformation from a fixed initial qubit state $\rho$ with a given probability $p$ iff
	\begin{subequations}
		\begin{align}
		&r^2 s_{z}^{2}+\left(1-r_{z}^{2}\right)s^2\le r^2, \label{eq:firstThmMain-2}\\
		&p^2s^2 \le \frac{r^2}{1+|r_z|}\left[2p-(1-|r_z|)\right]. \label{eq:secondThmMain-2}
		\end{align}
	\end{subequations}
\end{thmMain}

\begin{proof}
	According to Thm.~\ref{thm:equivSIOandIO}, we can focus on SIO transformations.
	In order to implement a stochastic
	qubit state transformation, we need a quantum instrument
	with two possible outcomes, success and failure, modelled by $\mathcal{E}_{s}^{\text{SIO}}(\rho)$
	and $\mathcal{E}_{f}^{\text{SIO}}(\rho)$. In the case of SIO transformations, both $\mathcal{E}_{s}^{\text{SIO}}(\rho)$
	and $\mathcal{E}_{f}^{\text{SIO}}(\rho)$ have to be decomposable
	into SIO Kraus operators. Due to Prop.~\ref{prop:compSIO}, we can
	focus exclusively on $\mathcal{E}_{s}^{\text{SIO}}$. 
	According to \cite{PhysRevLett.119.140402}, every $\mathcal{E}_{s}^{\text{SIO}}$
	can be represented by four SIO Kraus operators 
	\begin{align}
	K_{1} & =\begin{pmatrix}a_{1} & 0\\
	0 & b_{1}
	\end{pmatrix},K_{2}=\begin{pmatrix}0 & b_{2}\\
	a_{2} & 0
	\end{pmatrix},\nonumber \\
	K_{3} & =\begin{pmatrix}a_{3} & 0\\
	0 & 0
	\end{pmatrix},K_{4}=\begin{pmatrix}0 & b_{3}\\
	0 & 0
	\end{pmatrix}.
	\end{align}
	Since overall phases of Kraus operators are physically irrelevant,
	we assume from here on $a_{i},b_{3}\ge0$. Defining $\vec{a}=(a_{1},a_{2},a_{3})$
	and $\vec{b}=(b_{1},b_{2},b_{3})$, the condition that $\mathcal{E}_{s}^{\text{SIO}}$
	is trace non-increasing is equivalent to $l_{a}^2:=|\vec{a}|^{2}\le1$
	and $l_{b}^2:=|\vec{b}|^{2}\le1$. Due to symmetries and as explained
	in \cite{PhysRevLett.119.140402}, we can restrict our analysis to
	the case $s_{y}=r_{y}=0$ and $s_{x},r_{x},s_{z},r_{z}\ge0$. More
	precisely, we assume $r_{x}>0$ from here on, since otherwise we have
	the trivial case of incoherent initial states. From 
	\begin{align}
	\mathcal{E}_{s}^{\text{SIO}}(\rho)=p\sigma
	\end{align}
	then follow the Equations
	\begin{align}
	& ps_{x}=r_{x}\left(a_{2}\Real(b_{2})+a_{1}\Real(b_{1})\right),\nonumber \\
	& 0=a_{2}\Imag(b_{2})-a_{1}\Imag(b_{1}),\nonumber \\
	& p(1+s_{z})=\left(a_{1}^{2}+a_{3}^{2}\right)(1+r_{z})+\left(|b_{2}|^{2}+b_{3}^{2}\right)(1-r_{z}),\nonumber \\
	& p(1-s_{z})=a_{2}^{2}(1+r_{z})+|b_{1}|^{2}(1-r_{z})
	\end{align}
	or equivalently 
	\begin{alignat}{3}
	& ps_{x} &  & = &  & \ r_{x}\left(a_{2}\Real(b_{2})+a_{1}\Real(b_{1})\right),\nonumber \\
	& 0 &  & = &  & a_{2}\Imag(b_{2})-a_{1}\Imag(b_{1}),\nonumber \\
	& 2p &  & = &  & \ l_{a}^2(1+r_{z})+l_{b}^2(1-r_{z}),\nonumber \\
	& 2ps_{z} &  & = &  & \ \left(a_{1}^{2}+a_{3}^{2}-a_{2}^{2}\right)(1+r_{z})\nonumber \\
	&  &  &  &  & +\left(|b_{2}|^{2}+b_{3}^{2}-|b_{1}|^{2}\right)(1-r_{z}).
	\end{alignat}
	The principal idea of our proof from here on is the following: For fixed $r_{x},r_{z},p$, we determine states $(s_{x},s_{z})$ on
	the boundary of the region which is achievable with stochastic SIO,
	i.e. the region for which the Equations above have a solution for
	suitable $\vec{a},\vec{b}$. Since the achievable region is convex
	and contains the free states (we can always mix incoherently
	with a free state), this will allow us to deduce the entire reachable region.

	Now assume that $(s_{x},s_{z})$ is on the boundary of the reachable region. Then one can choose $a_{3}=0$
	and $b_{3}=0$, since $K_{3}$ and $K_{4}$ destroy all coherence.
	Formally, this can be shown considering 
	\begin{align}
	\vec{a}' & =(\sqrt{a_{1}^{2}+a_{3}^{2}},a_{2},0),\nonumber \\
	\vec{b}' & =(|b_{1}|,\sqrt{|b_{2}|^{2}+b_{3}^{2}},0),
	\end{align}
	which lead to 
	\begin{align}
	ps_{x}'= & \ r_{x}\left(a_{2}'\Real(b_{2}')+a_{1}'\Real(b_{1}')\right)\nonumber \\
	= & \ r_{x}\left(a_{2}\sqrt{|b_{2}|^{2}+b_{3}^{2}}+\sqrt{a_{1}^{2}+a_{3}^{2}}\ |b_{1}|\right)\nonumber \\
	\ge & \ ps_{x},\nonumber \\
	2ps_{z}'= & \ 2ps_{z},\nonumber \\
	l_{a'}^2= & \ l_{a}^2,\nonumber \\
	l_{b'}^2= & \ l_{b}^2,\nonumber \\
	0= & a'_{2}\Imag(b'_{2})-a'_{1}\Imag(b'_{1}).
	\end{align}
	Remember that we consider fixed $r_{x},r_{z}$ and $p>0$. Thus $s_{x}'\ge s_{x}$
	and $s_{z}'=s_{z}$. This mixing argument with the free states excludes boundaries of
	the achievable region parallel to the x-axis. Therefore $s_{x}'>s_{x}$
	for $s_{z}'=s_{z}$ cannot happen if both $(s_{x},s_{z})$ and $(s_{x}',s_{z}')$
	lie on the boundary and we will assume from here on $a_{3}=b_{3}=0$
	and $b_{1},b_{2}\ge0$. This leads to the Equations 
	\begin{alignat}{3} \label{eq:reducedEqSxSzP}
	& ps_{x} &  & = &  & \ r_{x}\left(a_{2}b_{2}+a_{1}b_{1}\right),\nonumber \\
	& 2p &  & = &  & \ l_{a}^2(1+r_{z})+l_{b}^2(1-r_{z}),\nonumber \\
	& 2ps_{z} &  & = &  & \ \left(a_{1}^{2}-a_{2}^{2}\right)(1+r_{z})+\left(b_{2}^{2}-b_{1}^{2}\right)(1-r_{z}).
	\end{alignat}
	Next we notice that the second line in the above Equations defines
	an ellipse. Remembering that we excluded the trivial case of $r_{z}=1$
	by assuming $r_{x}>0$, we can therefore use the parametrization 
	\begin{align}
	l_{a}= & \sqrt{\frac{2p}{1+r_{z}}}\cos\left(t\right),\nonumber \\
	l_{b}= & \sqrt{\frac{2p}{1-r_{z}}}\sin\left(t\right).
	\end{align}
	Without loss of generality, we choose $0\le t\le\pi/2$
	and the condition $l_{a},l_{b}\le1$ leads to 
	\begin{align} \label{eq:constraintsT}
	\cos\left(t\right)\le & \sqrt{\frac{1+r_{z}}{2p}},\nonumber \\
	\sin\left(t\right)\le & \sqrt{\frac{1-r_{z}}{2p}},
	\end{align}
	which restricts the range of $t$ further.
	Next we substitute 
	\begin{align}
	a_{1} & =\sqrt{\frac{2p}{1+r_{z}}}\cos\left(t\right)\cos\left(\frac{\theta-\phi}{2}\right),\nonumber \\
	a_{2} & =\sqrt{\frac{2p}{1+r_{z}}}\cos\left(t\right)\sin\left(\frac{\theta-\phi}{2}\right),\nonumber \\
	b_{1} & =\sqrt{\frac{2p}{1-r_{z}}}\sin\left(t\right)\sin\left(\frac{\theta+\phi}{2}\right),\nonumber \\
	b_{2} & =\sqrt{\frac{2p}{1-r_{z}}}\sin\left(t\right)\cos\left(\frac{\theta+\phi}{2}\right),
	\end{align}
	which automatically satisfies the ellipse Equation. Since all left hand sides of these Equations are positive by assumptions, we can choose without loss of generality $0\le\theta\le\pi/2$ and $-\theta\le\phi\le\theta (\Leftrightarrow 0\le \frac{\theta-\phi}{2},\frac{\theta+\phi}{2}\le \frac{\pi}{2})$.  The remaining
	two Equations are then (since $p>0$) 
	\begin{align} \label{eq:sxAndszByAngles}
	s_{x}= & \frac{r_{x}\sin(2t)\sin(\theta)}{\sqrt{1-r_{z}^{2}}},\nonumber \\
	s_{z}= & \cos(2t)\sin(\theta)\sin(\phi)+\cos(\theta)\cos(\phi).
	\end{align}
	When we know for every reachable $s_{x}$ the largest possible $s_{z}$, we achieved our goal of determining the boundary of the reachable region. 
	Therefore we
	fix $s_{x}$ and and maximize $s_{z}$. 
	For fixed $s_{x}$, we obtain
	from the first Equation a relation between $t$ and $\theta$, 
	\begin{align}
	\sin(\theta(t))=\frac{\sqrt{1-r_{z}^{2}}s_{x}}{r_{x}\sin(2t)}.
	\end{align}
	Using $0\le\theta\le\pi/2$, we can rewrite the second Equation as
	\begin{align*}
	s_{z}(t,\phi)= & \cos(2t)\sin(\theta(t))\sin(\phi)+\sqrt{1-\sin^{2}(\theta(t))}\cos(\phi),
	\end{align*}
	which is maximal either on the boundary or for
	\begin{align*}
	0= & \frac{\partial s_{z}(t,\phi)}{\partial\phi}\\
	= & \sin(\theta(t))\cos(2t)\cos(\phi)-\sqrt{1-\sin^{2}(\theta(t))}\sin(\phi).
	\end{align*}
	Since we have $-\pi/2\le-\theta\le\phi\le\theta\le\pi/2$, this is equivalent to
	\begin{align}
	\phi=\arctan\left(\frac{\sin(\theta(t))\cos(2t)}{\sqrt{(1-\sin^{2}(\theta(t)))}}\right).
	\end{align}
	Using that $\arctan(x)$ is monotonically increasing in $x$, we find
	\begin{align}
	\phi&\ge\arctan\left(\frac{-\sin(\theta)}{\sqrt{1-\sin(\theta)^2}}\right)=-\theta, \nonumber \\ \phi&\le\arctan\left(\frac{\sin(\theta)}{\sqrt{1-\sin(\theta)^2}}\right)=\theta
	\end{align}
	and therefore $\phi$ inside the allowed region.
	Then the $s_{z}(t)$, the $s_{z}$ optimized over $\phi$, is independent
	of $t$ and given by 
	\begin{align}\label{eq:szMax}
	s_{z}(t)=\sqrt{1-\frac{\left(1-r_{z}^{2}\right)s_{x}^{2}}{r_{x}^{2}}}.
	\end{align}
	Note that the expression under the square root is, due to Eq.~(\ref{eq:sxAndszByAngles}), never negative.
	
	Now we need to check the boundaries. 
	To do this, we express $t$ in terms of $\theta$ and define $X=\sin^2(\theta)$ ( therefore $(1-r_{z}^{2})s^2_{x}/r^2_{x}\le X \le 1$, again from Eq.~(\ref{eq:sxAndszByAngles})). For the moment, we assume $\cos(2t(\theta))\ge0$. This leads to
	\begin{align}
	s_z^+(\phi=\theta,\theta)=&\cos(2 t(\theta)) \sin^2 (\theta)+\cos^2(\theta) \nonumber \\
	=&\sqrt{1-\frac{(1-r_z^2)s_x^2}{r_x^2 \sin^2(\theta)}} \sin^2 (\theta)+\cos^2(\theta) \nonumber \\
	=&1-X+\sqrt{1-\frac{(1-r_z^2)s_x^2}{r_x^2 X}}X \nonumber \\
	=&s_z(X).
	\end{align}
	Since
	\begin{align}
	0=\frac{\partial}{\partial X} \left( 1-X+\sqrt{1-y/X}\ X \right)
	\end{align}
	has for $y\ne0$ no solutions, $s_z(X)$ attains its extrema on the boundaries. The exact maximum on the boundary depends on $t$, but it is lower than the maximum of
	\begin{align}
	s_z^+(X=(1-r_{z}^{2})s^2_{x}/r^2_{x})=&1-\frac{(1-r_z^2)s_x^2}{r_x^2}, \nonumber \\
	s_z^+(X=1)=&\sqrt{1-\frac{(1-r_z^2)s_x^2}{r_x^2}}.
	\end{align}
	and thus smaller than the extrema inside the allowed region. In the case of $\cos(2t(\theta))\le0$, we have
	\begin{align}
	s_z^-(\phi=\theta,\theta)=&\cos(2 t(\theta)) \sin^2 (\theta)+\cos^2(\theta) \nonumber \\
	=&-\sqrt{1-\frac{(1-r_z^2)s_x^2}{r_x^2 \sin^2(\theta)}} \sin^2 (\theta)+\cos^2(\theta) \nonumber \\
	=&1-X-\sqrt{1-\frac{(1-r_z^2)s_x^2}{r_x^2 X}}X \nonumber \\
	\le&s_z^+(\phi=\theta,\theta).
	\end{align}
	For the boundary with $\phi=-\theta$, the above considerations are the same, with the roles of $\cos(2t(\theta))\ge0$ and $\cos(2t(\theta))\le0$ inverted. 
	We thus confirmed that the maximal $s_z$ for given $s_x$ is indeed given by Eq.~(\ref{eq:szMax}) and independent of $\theta$ and $t$. 
	
	In order to finish the proof, we need to determine the reachable range of $s_x$ which depends according to Eq.~(\ref{eq:sxAndszByAngles}) on $t$ and therefore through Eqs.~(\ref{eq:constraintsT}) on $r_z$ and $p$. By the convexity of the reachable region, it is again sufficient to find the maximal reachable $s_x$. This corresponds to finding the allowed $t$ closest to $\pi/4$ (see again Eq.~(\ref{eq:sxAndszByAngles})), for which we will consider different cases. The first case is that neither of the conditions in Eq.~(\ref{eq:constraintsT}) restricts $t$, which is equivalent to 
	\begin{align}
	p\le \frac{1-r_z}{2}
	\end{align}
	and therefore
	\begin{align}
	s_x \le \frac{r_x}{\sqrt{1-r_z^2}}.
	\end{align}
	If
	\begin{align}
	p\le \frac{1+r_z}{2},
	\end{align} the constraints are 
	\begin{align}
	0\le t\le \arcsin \left(\sqrt{\frac{1-r_z}{2p}} \right).
	\end{align}
	For $p<1-r_z$, the upper bound on $t$ is larger than $\pi/4$, and we find the same bounds on $s_x$ as in the first case. Using
	\begin{align}
	\sin\left(2 \arcsin x \right)= 2x \sqrt{1-x^2},
	\end{align} we find
	\begin{align}
	s_x \le \frac{r_x}{\sqrt{1+r_z}}\frac{1}{p} \sqrt{2p-(1-r_z)}
	\end{align} otherwise. In the last case, for
	\begin{align}
	p\ge\frac{1+r_z}{2},
	\end{align}
	we have a lower and an upper bound on $t$,
	\begin{align}
	\arccos\left(\sqrt{\frac{1+r_z}{2 p}}\right)\le t\le \arcsin \left(\sqrt{\frac{1-r_z}{2p}} \right).
	\end{align}
	From Eq.~(\ref{eq:reducedEqSxSzP}), we see that the lower bound is always smaller than the upper. In addition, 
	\begin{align}
	\arccos\left(\sqrt{\frac{1+r_z}{2 p}}\right) \le \arccos\left(\frac{1}{\sqrt{2}}\right)=\frac{\pi}{4}.
	\end{align}
	Therefore, we end up with the same conclusions as in the second case.
	
	Finally, using the symmetry
	and mixing arguments, the reachable region is defined by the inequalities 
	\begin{align}
	&s_{z}^{2}\le1-\frac{1-r_{z}^{2}}{r_{x}^{2}+r_{y}^{2}}\left(s_{x}^{2}+s_{y}^{2}\right) \nonumber \\
	&\begin{cases}
	p< 1-|r_z|: &s_x^2+s_y^2 \le \frac{r_x^2+r_y^2}{1-r_z^2}  \\
	p\ge 1-|r_z|:&s_x^2+s_y^2 \le \frac{r_x^2+r_y^2}{1+|r_z|} \frac{1}{p^2}\left(2p-(1-|r_z|)\right)  \\
	\end{cases}
	\end{align}
	Rearranging the terms in the above Equations and using the short hand notations leads to 
	\begin{align} 
	&r^2 s_{z}^{2}+\left(1-r_{z}^{2}\right)s^2 \le r^2, \label{eqA:ellipse} \\
	&\begin{cases}
	p< 1-|r_z|: & \left(1-r_z^2\right)s^2 \le r^2,  \\
	p\ge 1-|r_z|:&p^2s^2 \le \frac{r^2}{1+|r_z|}\left(2p-(1-|r_z|)\right),  \\
	\end{cases}
	\end{align}
	formally also including the trivial cases of $r_{x}=r_{y}=0$. Now one can easily see that the condition for $p\le 1-|r_z|$ is always satisfied if condition (\ref{eqA:ellipse}) is satisfied.
	If we insert $p=1-|r_z|$ into the condition for $p\ge 1-|r_z|$, we obtain after simplifications
	\begin{align}
	(1-r_z^2) s^2\le r^2,
	\end{align}
	which is also always satisfied if condition~(\ref{eqA:ellipse}) is satisfied. Therefore the condition 
	\begin{align}
	p^2s^2 \le \frac{r^2}{1+|r_z|}\left(2p-(1-|r_z|)\right)
	\end{align}
	is for $p\le 1-|r_z|$ automatically satisfied, if condition~(\ref{eqA:ellipse}) holds.
	This leads us to the Theorem.
\end{proof}

As stated in the main text, these results allow us to evaluate the optimal conversion probability $P\left(\rho\rightarrow\sigma\right)$ via IO and SIO for any two states $\rho$ and $\sigma$ of a single qubit. We prove this in form of a Corollary.

\begin{cor}\label{cor:MaxProb}
	The maximal probability $P\left(\rho\rightarrow\sigma\right)$ 
	for a successful transformation from a coherent qubit state $\rho$
	to a coherent qubit state $\sigma$ using IO or SIO is zero if 
	\begin{align}\label{eq:CorProbZero}
	r^2 s_{z}^{2}+\left(1-r_{z}^{2}\right)s^2>r^2
	\end{align}
	and 
	\begin{align}\label{eq:CorProbability}
	\min\left\{ \frac{r^2}{\left(1+|r_z|\right)s^2}\left(1+\sqrt{1-\frac{s^2\left(1-r_z^2\right)}{r^2}}\right),1\right\} 
	\end{align}
	otherwise. 
\end{cor}
\begin{proof} 
	From Thm.~\ref{thm:Main} and its discussion in the Methods Section, we get that a transformation
	from $\rho$ to $\sigma$ (with $\rho$ coherent, i.e. $r>0$
	and therefore $r_{z}^{2}<1$) is possible with probability $p>0$
	iff 
	\begin{align}
	r^2 s_{z}^{2}+\left(1-r_{z}^{2}\right)s^2\le r^2.
	\end{align}
	As soon as we are inside this ellipsoid, the maximal probability of
	success is bounded by Eq.~(\ref{eq:secondThmMain-2}). Now we want to maximize $p$ such that this inequality is still satisfied. This is the case if we choose the larger $p$ for which 
	\begin{align}
	p^2s^2 = \frac{r^2}{1+|r_z|}\left(2p-(1-|r_z|)\right).
	\end{align}
	Together with the assumptions that $p_{\max}$ is
	a probability, this finishes the proof. 
\end{proof}

As announced in the main text, we show now how these results can be partially extended beyond qubits.  
Denoting by $C$ any coherence measure with the properties defined in~\cite{BaumgratzPhysRevLett.113.140401}, it holds that~\cite{Du:2015:CMO:2871378.2871381}
\begin{align}\label{eq:boundsTrafo}
P\left(\rho\rightarrow\sigma\right)  \le \frac{C(\rho)}{C(\sigma)},
\end{align}
which we prove now for completeness.
Every stochastic coherence transformation from $\rho$ to $\sigma$ can be described by an incoherent quantum instrument with two possible outcomes, success and failure. We denote by $K_n$ the incoherent Kraus operators modelling the case of success and by $L_m$ the ones describing the event of failure. With
\begin{align}
p_n=&\tr\left(K_n \rho K_n^\dagger\right), \nonumber \\
q_m=&\tr\left(L_m \rho L_m^\dagger \right), \nonumber \\
\sigma_n=&K_n \rho K_n^\dagger/p_n, \nonumber \\
\chi_m=&L_m \rho L_m^\dagger/q_m, \nonumber \\
P\left(\rho\rightarrow\sigma\right)=&\sum_n p_n, \nonumber \\
q=&\sum_m q_m, 
\end{align}
we first use property (C2b), then (C3) and finally (C1) defined in \cite{BaumgratzPhysRevLett.113.140401} to arrive at
\begin{align}
C(\rho)\ge& \sum_n p_n C(\sigma_n)+ \sum_m q_m C(\chi_m) \nonumber \\
=&P\left(\rho\rightarrow\sigma\right) \sum_n \frac{p_n}{P\left(\rho\rightarrow\sigma\right)} C(\sigma_n)+ q\sum_m \frac{q_m}{q} C(\chi_m) \nonumber \\
\ge&P\left(\rho\rightarrow\sigma\right) C\left(\sum_n \frac{p_n}{P\left(\rho\rightarrow\sigma\right)} \sigma_n\right)+ qC\left(\sum_m \frac{q_m}{q} \chi_m\right) \nonumber \\
\ge&P\left(\rho\rightarrow\sigma\right) C\left(\sigma\right).
\end{align}

Note that the bounds given in Eq.~(\ref{eq:boundsTrafo}) cannot be used to exclude the existence of a stochastic transformation from $\rho$ to $\sigma$ (unless we have the trivial cases $C(\sigma)=\infty$ or $\rho$ incoherent and $\sigma$ not). However, the first condition in Cor.~\ref{cor:MaxProb} is a (nontrivial) necessary condition for the existence of a stochastic transformation. In the case of SIO, we can generalize this necessary condition to arbitrary dimensions using the $\Delta$~robustness of coherence $C_{\Delta,R}$ introduced in~\cite{ChitambarPhysRevLett.117.030401,ChitambarPhysRevA.94.052336} by
\begin{align}
C_{\Delta,R}(\rho)=&\min \Bigg\{ t\ge0\Bigg|\frac{\rho+t\tau}{1+t}\in \mathcal{I},\tau\ge0,\Delta\rho=\Delta \tau \Bigg\},
\end{align}
where $\mathcal{I}$ denotes the set of incoherent states.
\begin{thm}\label{thm:necessaryStochSIO}
	A necessary condition for the existence of a stochastic SIO transformation from $\rho$ to $\sigma$ is
	\begin{equation}\label{eq:necessarySIO}
	C_{\Delta,R}(\sigma)\le C_{\Delta,R}(\rho).
	\end{equation} 
\end{thm}

\begin{proof}
	Assume that there exists a SIO transformation which maps $\rho$ to $\sigma$ with probability $p\ne0$. Due to Prop.~\ref{prop:compSIO}, we can restrict ourselves to the case of successful transformations which we write as,
	\begin{align}
	\Lambda [\rho] = p\sigma.
	\end{align}
	From here on, we will suppress unnecessary brackets for readability and write for example $\Lambda\rho$ instead of $\Lambda [\rho]$.
	Since $\Lambda$ can be decomposed into strictly incoherent Kraus operators, we have $\Delta \Lambda=\Lambda \Delta$ and therefore
	\begin{align}
	\Lambda \Delta \rho =\Delta \Lambda \rho =p\Delta \sigma.
	\end{align} 
	Defining 
	\begin{align}
	s_\tau = \max \{s| \Delta \tau +s ( \Delta \tau-  \tau) \ge 0\},
	\end{align}
	ensures that 
	\begin{align}
	\tilde{\rho}=\Delta \rho +s_\rho (\Delta\rho- \rho)
	\end{align}
	is a valid density operator with the property
	\begin{align}
	\Lambda \tilde{\rho}= p (\Delta \sigma +s_\rho(\Delta\sigma- \sigma))=:p\tilde{\sigma},
	\end{align}
	where $\tilde{\sigma}$ is a valid density operator too. This implies that
	\begin{align}\label{eq:necSIOTrafo}
	s_\sigma\ge s_\rho
	\end{align} 
	is a necessary condition for the existence of a stochastic transformation from $\rho$ to $\sigma$ using SIO. From the definition of $s_\tau$ follows directly that this quantity can be calculated using a semidefinite program. In addition, we will show now how Eq.~(\ref{eq:necSIOTrafo}) can be reformulated in terms of the $\Delta$~robustness of coherence $C_{\Delta,R}$~\cite{ChitambarPhysRevLett.117.030401,ChitambarPhysRevA.94.052336}
	\begin{align}
	C_{\Delta,R}(\rho)=&\min \Bigg\{ t\ge0\Bigg|\frac{\rho+t\tau}{1+t}\in \mathcal{I},\tau\ge0,\Delta\rho=\Delta \tau \Bigg\}.
	\end{align}
	Because $\Delta \rho=\Delta \tau$, this is equivalent to
	\begin{align*}
	C_{\Delta,R}(\rho)=&\min \Big\{ t\ge0\Big|\rho=(1+t)\Delta \rho-t\tau,\tau\ge0,\Delta\rho=\Delta \tau \Big\}.
	\end{align*}
	Using the same technique as in~\cite{ProofRobustness}, we can further simplify this expression to
	\begin{align} \label{eq:deltaRobSimp}
	C_{\Delta,R}(\rho)=&\min \Big\{ t\ge0\Big|\rho\le(1+t)\Delta \rho \Big\}.
	\end{align}
	To prove this, we first note that for $t,\tau\ge0$, 
	\begin{align}
	\rho=(1+t)\Delta \rho-t\tau
	\end{align}
	implies $\rho \le (1+t)\Delta \rho$. To show the converse, assume that $\rho \le (1+t)\Delta \rho$. Then we can define $\tau:=\left[(1+t)\Delta\rho-\rho\right]/t$ and it is easy to check that $\Delta \tau =\Delta \rho,\ \rho=(1+t)\Delta \rho-t\tau$ and $\tau\ge0$.
	Substituting $s$ by $1/t$, we find
	\begin{align}
	s_\tau =& \max \{s| \Delta \tau +s ( \Delta\tau-  \tau) \ge 0\} \nonumber \\
	=& \max \{s\ge0| \Delta \tau +s ( \Delta \tau- \tau) \ge 0\} \nonumber \\
	=& \max \{1/t| \Delta \tau + (\Delta \tau-  \tau)/t \ge 0,t\ge 0\} \nonumber \\
	=& \max \{1/t| \tau\le(1+t)\Delta \tau,t\ge 0\}.
	\end{align}
	A comparison with Eq.~(\ref{eq:deltaRobSimp}) shows that $C_{\Delta,R}(\tau)=1/s_\tau$, and therefore 
	Eq.~(\ref{eq:necSIOTrafo}) is equivalent to 
	\begin{align}
	C_{\Delta,R}(\sigma)\le C_{\Delta,R}(\rho).
	\end{align}
\end{proof}

As shown in~\cite{ChitambarPhysRevA.94.052336}, for the case of qubits, Eq.~(\ref{eq:necessarySIO}) is equivalent to conditions~(\ref{eq:firstThmMain-2}) and~(\ref{eq:CorProbZero}) and for higher dimensions, $C_{\Delta,R}$ can be evaluated efficiently using a semidefinite program (see the above proof and also~\cite{ProofRobustness}). The other necessary condition for stochastic transformations on qubits was that the initial state is not incoherent. For higher dimensions, this can be generalized by the statement that the coherence rank or number~\cite{Nathan,WinterPhysRevLett.116.120404,Chin1,Chin2} can only decrease under a stochastic IO (and therefore SIO) transformation, which we show now for completeness.

The coherence rank $r_C$ of pure states is defined as the number of non-zero coefficients needed to expand the state in the incoherent basis~\cite{Nathan,WinterPhysRevLett.116.120404}. For mixed states, the coherence rank is defined by~\cite{Chin1}
\begin{align}
r_C(\rho)=	\min\Big\{\max_i r_C(\ket{\psi_i}) \Big| \rho=\sum_i p_i \ketbra{\psi_i}{\psi_i}, p_i\ge0 \Big\}.
\end{align}
It is well known that the coherence rank of a pure state can only decrease under the action of an incoherent Kraus operator~\cite{WinterPhysRevLett.116.120404}. From this follows the statement directly: Let $\{p_i,\ket{\psi_i}\}$ be an optimal decomposition of $\rho$ in the sense that $r_C(\rho)=\max_i r_C(\ket{\psi_i})$. Applying the Kraus operators of the stochastic incoherent operation to the $\ket{\psi_i}$ leads to a decomposition of the final state with the promised property.

Now we turn to the proofs of our results concerning optimal conversion with assistance. We begin with the poof of Theorem 2 from the main text.

\begin{thmMain}
	\label{thm:PurificationAssistance}
	Let Alice and Bob share a pure two-qubit state  $\ket{\psi}^{AB}$ and denote Bob's local state by $\rho^B$. The maximal probability  $P_\mathrm{a}(\ket{\psi}^{AB}\rightarrow\sigma^{B})$ to prepare the qubit state $\sigma^B$ on Bob's side via one-way LQICC is given by
	
	\begin{equation}
	P_\mathrm{a}\left(\ket{\psi}^{AB}\rightarrow\sigma^{B}\right)=
	\min\left\{ 1,\left(1-\left|r_z\right|\right)\frac{1+\sqrt{1-s^2}}{s^2}\right\},\label{eq:PurificationAssistance-2}
	\end{equation}
	where $\vec{r}$ and $\vec{s}$ are the Bloch vectors of $\rho^B$ and $\sigma^B$, respectively.
\end{thmMain}	
\begin{proof}
	In the following, we will prove the more general
	case in which Alice holds an arbitrary purification of
	Bob's qubit state. Recall that by performing local measurements on Alice's side and using
	classical communication, Bob can obtain any decomposition
	$\{q_{i},\rho_i^{B}\}$ of his local state $\rho^{B}=\sum_{i}q_{i}\rho_i^{B}$~\cite{HUGHSTON199314}.
	After Alice's measurement has been performed, Bob applies incoherent
	operations to stochastically transform his post-measurement states
	$\rho_i^{B}$ into the desired state $\sigma^{B}$. First we will show that there always exists an optimal decomposition containing only pure states. 
	
	To do this, we use a decomposition of  
	every mixed qubit state $\rho$ into two pure states, which we will also use later in this proof. The Bloch vector $\boldsymbol{r}$ corresponding to $\rho$ can be written as a convex combination
	of two points $\boldsymbol{t}$ and $\boldsymbol{u}$ on the surface
	of the Bloch sphere having the same $z$-coordinate as $\boldsymbol{r}$,
	i.e., 
	\begin{subequations}\label{eq:decomposition-2}
		\begin{align}
		\boldsymbol{r} & =q\boldsymbol{t}+(1-q)\boldsymbol{u},\\
		r_{z} & =t_{z}=u_{z},\\
		|\boldsymbol{t}| & =|\boldsymbol{u}|=1.
		\end{align}
	\end{subequations}
	
	Now assume that an optimal decomposition of Bob's local state contains a mixed state $\rho_x$ which occurs with probability $q_x$. Since every decomposition can be obtained, Bob can also obtain a decomposition in which the pair $\{q_x,\rho_x\}$ is replaced by the two corresponding pure states from Eqs.~(\ref{eq:decomposition-2}) and probabilities $q_xq,q_x(1-q)$. 
	From Cor.~\ref{cor:MaxProb} follows that the transformation probability from both of these states to any target state is at least as high as the probability from $\rho_x$: In case the transformation from $\rho_x$ to a target state is forbidden by Eq.~(\ref{eq:CorProbZero}), there is nothing to prove. If not, the transformation probability from $\rho_x$ and the two pure states to the target is determined by Eq.~(\ref{eq:CorProbability}), since a pure initial state can never satisfy Eq.~(\ref{eq:CorProbZero}). Remember that by choice, we have $r_z=t_z=u_z$ and $t,u\ge r$. Now note that for a fixed target state $\sigma$ (fixed $\boldsymbol{s}$) and fixed $r_z$, the quantity in Eq.~(\ref{eq:CorProbability}) increases if $r$ increases. From this follows the claim, which implies that the new decomposition is also optimal. Eliminating all mixed states using this procedure, we end up with an optimal decomposition $\{q_{i},\ket{\psi_{i}}^{B}\}$ which only contains pure states.

	If
	we denote the corresponding maximal conversion probability by $P_\mathrm{a}\left(\ket{\psi_{i}}^{B}\rightarrow\sigma^{B} \right) $,
	our figure of merit can be written as 
	\begin{equation}
	P_\mathrm{a}\left(\ket{\psi}^{AB}\rightarrow\sigma^{B}\right)=\max\sum_{i}q_{i}P\left(\ket{\psi_{i}}^{B}\rightarrow\sigma^{B}\right),\label{eq:ConcaveRoof}
	\end{equation}
	where the maximization is performed over all pure state decompositions
	$\{q_{i},\ket{\psi_{i}}^{B}\}$ of Bob's local state.
	
	In the next step, we deduce from Eq.~(\ref{eq:CorProbability}) that the maximal probability for stochastic
	conversion between two single-qubit states $\ket{\psi_{i}}^{B}$
	and $\sigma^{B}$ can be expressed as
	\begin{equation}
	P\left(\ket{\psi_{i}}^{B}\rightarrow\sigma^{B}\right)=
	\min\left\{1, \left(1-|r_z^i|\right)\frac{1+\sqrt{1-s^2}}{s^2}\right\}
	.\label{eq:QubitPure}
	\end{equation}
	This result allows us to bound the average conversion probability
	for a decomposition $\{q_{i},\ket{\psi_{i}}^{B}\}$ of the state $\rho^{B}$
	as follows:
	\begin{align}
	\sum_{i}q_{i}P&\left(\ket{\psi_{i}}^{B}\rightarrow\sigma^{B}\right) \nonumber\\
	& =\sum_{i}q_{i}\min\left\{1, \left(1-|r_z^i|\right)\frac{1+\sqrt{1-s^2}}{s^2}\right\} \nonumber \\
	& \leq\sum_{i}q_{i}\times \left(1-|r_z^i|\right)\frac{1+\sqrt{1-s^2}}{s^2}\nonumber \\
	& =\left(1-\sum_{i}q_{i}|r_z^i|\right)\frac{1+\sqrt{1-s^2}}{s^2}\nonumber \\
	& \leq\left(1-|r_z|\right)\frac{1+\sqrt{1-s^2}}{s^2},
	\label{eq:EnsembleBound}
	\end{align}
	where in the last inequality we used the fact that $\sum_{i}q_{i}|r_z^i|=\sum_{i}q_{i}|\!\braket{\psi_{i}|\sigma_{z}}{\psi_{i}}\!|\geq\left|\mathrm{Tr}[\rho^{B}\sigma_{z}]\right|=|r_z|$.
	Since Eq.~(\ref{eq:EnsembleBound}) holds for any decomposition $\{q_{i},\ket{\psi_{i}}^{B}\}$
	of the state $\rho^{B}$, it implies that the probability for assisted
	conversion $P_\mathrm{a}\left(\ket{\psi}^{AB}\rightarrow\sigma^{B}\right)$ is bounded
	as 
	\begin{equation}
	P_\mathrm{a}\left(\ket{\psi}^{AB}\rightarrow\sigma^{B}\right)\leq\min\left\{ 1,\left(1-|r_z|\right)\frac{1+\sqrt{1-s^2}}{s^2}\right\} .\label{eq:Bound}
	\end{equation}

	To complete the proof of the Theorem, let Alice perform a two-outcome
	measurement such that Bob's post-measurement states $\ket{\psi_{i}}^{B}$ 
	fulfill 
	\begin{align}
	\mathrm{Tr}[\rho^{B}\sigma_{z}] & =\braket{\psi_{1}|\sigma_{z}}{\psi_{1}}=\braket{\psi_{2}|\sigma_{z}}{\psi_{2}},\label{eq:QubitOptimal-2}
	\end{align}
	Such a measurement always exists, see discussion above Eqs.~(\ref{eq:decomposition-2}).
	
	Depending on the outcome
	$i$, Bob then applies a stochastic incoherent operation to convert
	the state $\ket{\psi_{i}}^{B}$ into the desired state $\sigma^{B}$.
	This conversion protocol gives a lower bound on our figure of merit:
	\begin{align}
	P_\mathrm{a}\left(\ket{\psi}^{AB}\rightarrow\sigma^{B}\right)  \geq& qP\left(\ket{\psi_{1}}^{B}\rightarrow\sigma^{B}\right)\nonumber \\
	&+[1-q]P\left(\ket{\psi_{2}}^{B}\rightarrow\sigma^{B}\right)\nonumber \\
	=&\min\left\{1, \left(1-|r_z|\right)\frac{1+\sqrt{1-s^2}}{s^2}\right\} ,
	\label{eq:QubitOptimal-3}
	\end{align}
	where in the last equality we used Eqs.~(\ref{eq:QubitPure}) and
	(\ref{eq:QubitOptimal-2}). Noting that the lower bound~(\ref{eq:QubitOptimal-3})
	coincides with the upper bound~(\ref{eq:Bound}), we conclude that
	the presented protocol achieves the claimed conversion probability~(\ref{eq:PurificationAssistance-2})
	and that this probability is optimal. 
\end{proof}

Next we prove our results concerning two-qubit Werner states.

\begin{thm}
	The optimal probability $P_\mathrm{a}\left(\rho_{\mathrm{w}}^{AB}\rightarrow\sigma^{B}\right)$ for converting $\rho_{\mathrm{w}}^{AB}$ into the qubit state $\sigma^B$ via one-way LQICC is given by,
	\begin{equation} \label{eq:Werner-3}
	P_\mathrm{a}\left(\rho_{\mathrm{w}}^{AB}\rightarrow\sigma^{B}\right)=\begin{cases}
	1 & \mathrm{if}\,\,\,q_\mathrm{w}\geq \frac{s^2}{\sqrt{1-s_z^2}},\\  
	0 & \mathrm{otherwise}.
	\end{cases}
	\end{equation}
\end{thm}
\begin{proof}
	Suppose that Alice performs a general local measurement with POVM
	elements $\{M_{i}^{A}\}$. Conditioned on the measurement outcome
	$i$, Bob finds his system in the state $\rho_{i}^{B}=\mathrm{Tr}_{A}[M_{i}^{A}\rho_{\mathrm{w}}^{AB}]/p_{i}$,
	where $p_{i}=\mathrm{Tr}[M_{i}^{A}\rho_{\mathrm{w}}^{AB}]$ is the
	corresponding probability. To determine the set of states that Bob
	can achieve in this setting with nonzero probability, recall from the discussion below Thm.~\ref{thm:necessaryStochSIO} that
	Bob can transform a qubit state  $\rho$ into another qubit state $\tilde{\rho}$
	via stochastic incoherent operations if and only if
	\begin{equation}
	C_{\Delta,R}(\rho)\geq C_{\Delta,R}(\tilde{\rho}),
	\end{equation}
	where $C_{\Delta,R}$ is the $\Delta$-robustness of coherence~\cite{ChitambarPhysRevLett.117.030401,ChitambarPhysRevA.94.052336,ChitambarPhysRevA.95.019902}, which,
	for a single-qubit state $\rho$, can be expressed as
	\begin{align}
	C_{\Delta,R}(\rho)=\frac{r}{\sqrt{1-r_z^2}}.
	\end{align}	
	Thus, for characterizing the set of states achievable with nonzero
	probability, we need to evaluate the maximal $\Delta$-robustness
	$C_{\Delta,R}$ for any possible post-measurement state of Bob. Noting
	that $C_{\Delta,R}$ is convex, which follows directly from Eq.~(\ref{eq:deltaRobSimp}), we
	can restrict ourselves to rank-one POVMs on Alice's side.
	
	In the next step, note that for any rank one POVM element $M^{A}$,
	the corresponding post-measurement state of Bob has the form 
	\begin{equation}
	\rho^{B}=q_\mathrm{w}\ket{\eta}\!\bra{\eta}^{B}+(1-q_\mathrm{w})\frac{\openone^{B}}{2}.
	\end{equation}
	While the probability $q_\mathrm{w}$ here is fixed via the initial Werner state
	\begin{equation}
	\rho_{\mathrm{w}}^{AB}=q_\mathrm{w}\ket{\phi^{+}}\!\bra{\phi^{+}}+(1-q_\mathrm{w})\frac{\openone}{4},\label{eq:Werner-1}
	\end{equation}
	we can arbitrarily vary the state $\ket{\eta}$ by suitable adjusting
	Alice's POVM elements. Among all such states, the maximal $\Delta$-robustness
	$C_{\Delta,R}$ is attained for 
	\begin{equation}
	\mu^{B}=q_\mathrm{w}\ket{+}\!\bra{+}^{B}+(1-q_\mathrm{w})\frac{\openone^{B}}{2},
	\end{equation}
	i.e. it holds $C_{\Delta,R}(\rho^{B})\leq C_{\Delta,R}(\mu^{B})$. To
	see this, note that for single-qubit states the $\Delta$-robustness
	does not increase under incoherent operations~\cite{ChitambarPhysRevLett.117.030401,ChitambarPhysRevA.94.052336,ChitambarPhysRevA.95.019902},
	and moreover for any $\ket{\eta}^{B}$ there exists an incoherent
	operation converting $\mu^{B}$ into $\rho^{B}$~\cite{Streltsov1367-2630-20-5-053058}.
	Combining these arguments, we see that any post-measurement state
	of Bob has not more $\Delta$-robustness of coherence than the state
	$\mu^{B}$. Thus, we obtain the following condition for assisted state
	conversion of the Werner state~(\ref{eq:Werner-1}) into a state
	$\sigma^{B}$ on Bob's side:
	\begin{equation}
	P_\mathrm{a}\left(\rho_{\mathrm{w}}^{AB}\rightarrow\sigma^{B}\right)>0\,\,\,\Rightarrow\,\,\,C_{\Delta,R}(\sigma^{B})\leq C_{\Delta,R}(\mu^{B}).\label{eq:Condition-1}
	\end{equation}
	
	We will now show that the state $\mu^{B}$ is in fact achievable with
	unit probability: $P_\mathrm{a}\left(\rho_{\mathrm{w}}^{AB}\rightarrow\mu^{B}\right)=1$.
	For this, Alice first performs a local measurement on the Werner state
	in the $\{\ket{+},\ket{-}\}$ basis. Conditioned on the measurement
	outcome Bob finds his system either in the desired state $\mu^{B}$,
	or in the state $\sigma_{z}\mu^{B}\sigma_{z}$. In the latter case,
	Bob performs an incoherent $\sigma_{z}$ rotation, thus obtaining
	$\mu^{B}$ with unit probability. Note that via local incoherent operations
	Bob can transform the state $\mu^{B}$ into another state $\sigma^{B}$
	with unit probability if and only if $C_{\Delta,R}(\sigma^{B})\leq C_{\Delta,R}(\mu^{B})$. This follows from the discussion of the geometrical interpretation of Thm.~\ref{thm:Main} in the Methods Section of the main text. All states reachable from $\mu^B$ with non-zero probability are inside the ellipse through $\mu^B$, which is equivalent to $C_{\Delta,R}(\sigma^{B})\leq C_{\Delta,R}(\mu^{B})$. These states are  reachable with certainty, since $\mathrm{Tr}(\mu^B\sigma_z)=0$. Combining these arguments, we obtain
	the following condition:
	\begin{equation}
	C_{\Delta,R}(\sigma^{B})\leq C_{\Delta,R}(\mu^{B})\,\,\,\Rightarrow\,\,\,P_\mathrm{a}\left(\rho_{\mathrm{w}}^{AB}\rightarrow\sigma^{B}\right)=1.\label{eq:Condition-2}
	\end{equation}
	Both conditions~(\ref{eq:Condition-1}) and~(\ref{eq:Condition-2})
	imply the following: \begin{subequations}
		\begin{align}
		C_{\Delta,R}(\sigma^{B}) & \leq C_{\Delta,R}(\mu^{B})\,\,\,\Leftrightarrow\,\,\,P_\mathrm{a}\left(\rho_{\mathrm{w}}^{AB}\rightarrow\sigma^{B}\right)=1,\\
		C_{\Delta,R}(\sigma^{B}) & >C_{\Delta,R}(\mu^{B})\,\,\,\Leftrightarrow\,\,\,P_\mathrm{a}\left(\rho_{\mathrm{w}}^{AB}\rightarrow\sigma^{B}\right)=0.
		\end{align}
	\end{subequations}The proof of the Theorem is complete by noting
	that $C_{\Delta,R}(\mu^{B})=q_\mathrm{w}$.
\end{proof}

We will finish this section with the proof of our claim from the main text that correlations enhance conversion probabilities, as long as the global state is not quantum-incoherent. 

\begin{thm}
	If Bob's system is a qubit, then for any state $\rho^{AB}$ which
	is correlated and not quantum-incoherent the set of accessible states
	for Bob via stochastic one-way LQICC is strictly larger, when compared to
	$\rho^{A}\otimes\rho^{B}$.
\end{thm}
\begin{proof}
	As is shown in Lem.~\ref{lem:2}, for any correlated state $\rho^{AB}\neq\rho^{A}\otimes\rho^{B}$
	there exists a two-element POVM $\{M_{1}^{A},M_{2}^{A}\}$ on Alice's
	side such that $\rho_{1}^{B}\neq\rho_{2}^{B}$, where $\rho_{i}^{B}$
	is the state of Bob conditioned on the measurement outcome of Alice:
	\begin{equation}
	\rho_{i}^{B}=\frac{\mathrm{Tr}_{A}\left[M_{i}^{A}\rho^{AB}\right]}{p_{i}},
	\end{equation}
	and $p_{i}=\mathrm{Tr}\left[M_{i}^{A}\rho^{AB}\right]$ is the corresponding
	probability for obtaining the outcome $i$. Noting that 
	\begin{equation}
	\rho^{B}=p_{1}\rho_{1}^{B}+p_{2}\rho_{2}^{B},
	\end{equation}
	we conclude that \textendash{} whenever the state $\rho^{AB}$ is
	not quantum-incoherent \textendash{} either $\rho_{1}^{B}$ or $\rho_{2}^{B}$
	must be outside of the reachable ellipsoid of Bob's reduced state
	$\rho^{B}$. This completes the proof of the Theorem.
\end{proof}

\section{Implications of the theoretical results}
In this section, we discuss the implication of our theoretical results, which we briefly mentioned in the discussion of the main text.

In the scenario
considered so far we assumed that incoherent operations are applied
on one copy of the state $\rho$. In the following we will extend
our investigations to asymptotic conversion scenarios, where incoherent
operations are performed on a large number of copies of the state
$\rho$. The figure of merit in this setting is the asymptotic conversion
rate 
\begin{equation*}
R(\rho\rightarrow\sigma)=\sup\left\{ r:\lim_{n\rightarrow\infty}\left(\inf_{\Lambda}\left\Vert \Lambda\left(\rho^{\otimes n}\right)-\sigma^{\otimes\left\lfloor rn\right\rfloor }\right\Vert _{1}\right)=0\right\} ,\label{eq:R}
\end{equation*}
where $||M||_{1}=\mathrm{Tr}\sqrt{M^{\dagger}M}$ is the trace norm,
the infimum is performed over all incoherent operations $\Lambda$,
and $\left\lfloor x\right\rfloor $ is the largest integer smaller
or equal to the real number $x$.

It is now important to note that the single copy conversion probability
$P(\rho\rightarrow\sigma)$ is a lower bound for the conversion rate:
\begin{equation}
R(\rho\rightarrow\sigma)\geq P(\rho\rightarrow\sigma).\label{eq:bound-1}
\end{equation}
In fact, asymptotic conversion at rate $P(\rho\rightarrow\sigma)$
can be achieved by applying stochastic IO on each copy of the state
$\rho$. Denoting by $C_{\mathrm{d}}$ the distillable coherence and by $C_{\mathrm{c}}$ the coherence cost~\cite{WinterPhysRevLett.116.120404}, the bounds
\begin{equation}
\frac{C_{\mathrm{d}}(\rho)}{C_{\mathrm{c}}(\sigma)} \leq R(\rho\rightarrow\sigma)\leq\min\left\{ \frac{C_{\mathrm{d}}(\rho)}{C_{\mathrm{d}}(\sigma)},\frac{C_{\mathrm{c}}(\rho)}{C_{\mathrm{c}}(\sigma)}\right\} .\label{eq:bound-2}
\end{equation}
appeared in \cite{WinterPhysRevLett.116.120404}.

As was shown again in~\cite{WinterPhysRevLett.116.120404}, the distillable
coherence admits the following closed expression: 
\begin{equation}\label{eq:distillableClosed}
C_{\mathrm{d}}(\rho)=S(\Delta[\rho])-S(\rho),
\end{equation}
where $S(\rho)=-\mathrm{Tr}[\rho\log_{2}\rho]$ is the von Neumann
entropy and $\Delta[\rho]=\sum_{i}\ketbra{i}{i}\rho\ketbra{i}{i}$
is the dephasing operator. Moreover, the coherence cost $C_{\mathrm{c}}$
is equal to the coherence of formation $C_{\mathrm{f}}$~\cite{WinterPhysRevLett.116.120404}:
\begin{equation} \label{eq:coherenceFormation}
C_{\mathrm{c}}(\rho)=C_{\mathrm{f}}(\rho)=\min\sum_{i}p_{i}S\left(\Delta\left[\psi_{i}\right]\right).
\end{equation}
Here, the minimization is performed over all pure state decompositions
of the state $\rho=\sum_{i}p_{i}\psi_{i}$. 

Up until here, the results concerning asymptotic conversions were valid for general dimensions. From here on, we will specialize them exclusively to qubits.
For single-qubit states, Eq.~(\ref{eq:coherenceFormation}) can be further simplified as follows~\cite{YuanPhysRevA.92.022124}:
\begin{equation}
C_{\mathrm{c}}(\rho)=C_{\mathrm{f}}(\rho)=h\left(\frac{1+\sqrt{1-4|\rho_{01}|^{2}}}{2}\right),\label{eq:CcQubit}
\end{equation}
where $h(x)=-x\log_{2}x-(1-x)\log_{2}(1-x)$ is the binary entropy
and $\rho_{01}=\braket{0}{\rho|1}$.

\begin{figure}
	\includegraphics[height=4cm]{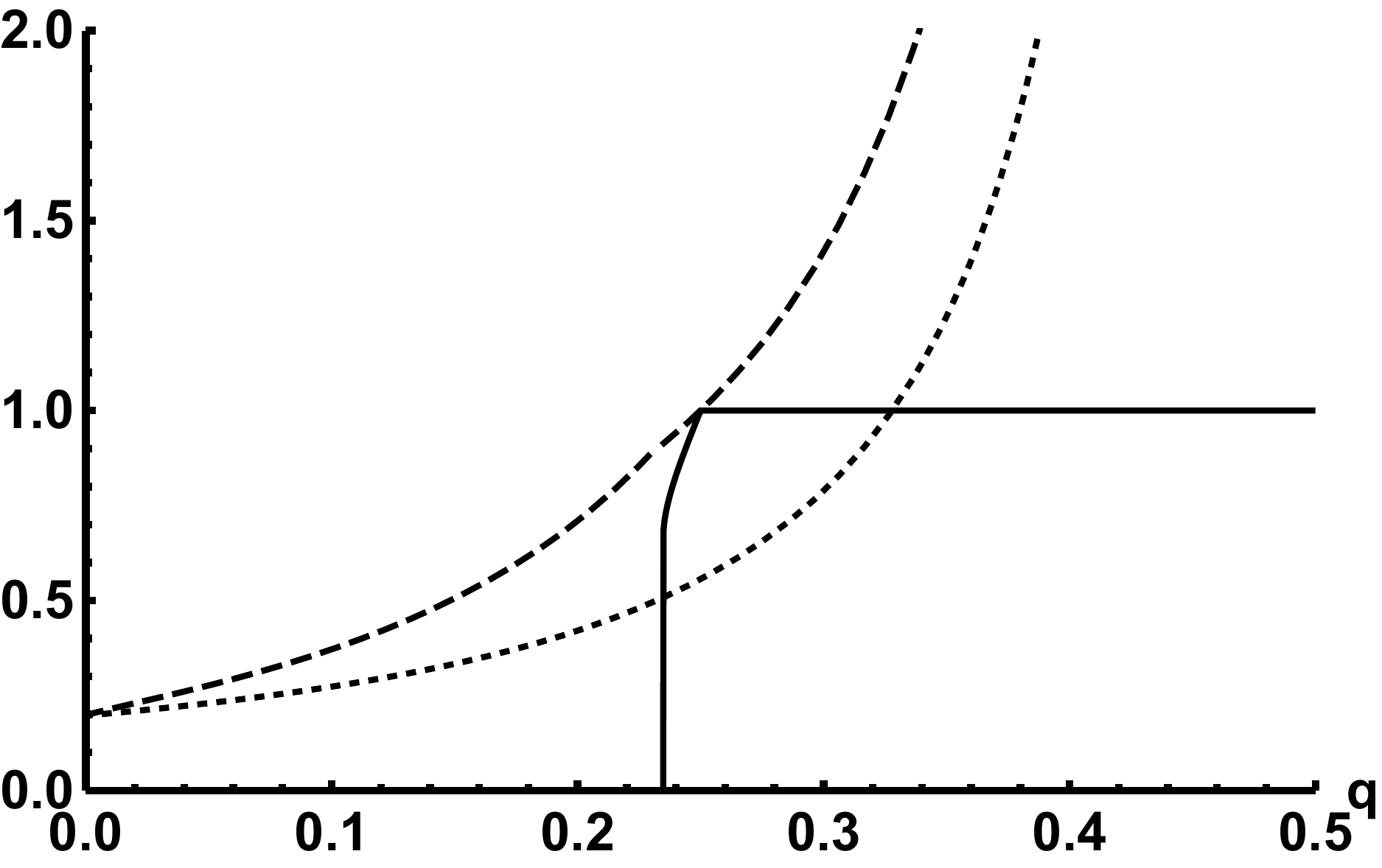}
	
	\caption{\label{fig:bounds}Comparison of upper and lower bounds on the asymptotic
		conversion rate $R(\rho\rightarrow\sigma)$ for states in Eqs.~(\ref{eq:ExampleRho})
		and (\ref{eq:ExampleSigma}). Dashed line shows the upper bound given
		by $\min\left\{ C_{\mathrm{d}}(\rho)/C_{\mathrm{d}}(\sigma),C_{\mathrm{c}}(\rho)/C_{\mathrm{c}}(\sigma)\right\} $,
		solid line shows the lower bound given by $P(\rho\rightarrow\sigma)$,
		and dotted line shows the lower bound given by $C_{\mathrm{d}}(\rho)/C_{\mathrm{c}}(\sigma)$.}
\end{figure}

We will now demonstrate the power of these results on a specific example.
For this, we consider the following single-qubit state: 
\begin{align}
\rho & =\left(\begin{array}{cc}
\frac{2}{3} & \frac{1}{4}\\
\frac{1}{4} & \frac{1}{3}
\end{array}\right).\label{eq:ExampleRho}
\end{align}
We will study the conversion of $\rho$ into a convex combination
of maximally coherent states $\ket{\pm}=(\ket{0}\pm\ket{1})/\sqrt{2}$,
i.e., the final state $\sigma$ has the form 
\begin{equation}
\sigma=q\ket{+}\!\bra{+}+(1-q)\ket{-}\!\bra{-}.\label{eq:ExampleSigma}
\end{equation}
In Fig.~\ref{fig:bounds} we compare the aforementioned upper and
lower bounds on the state conversion rate for the states $\rho$ and
$\sigma$ in Eqs.~(\ref{eq:ExampleRho}) and~(\ref{eq:ExampleSigma}).
In particular, there exists a range of the parameter $q$ where 
$P(\rho\rightarrow\sigma)$ {[}solid line
in Fig.~\ref{fig:bounds}{]} is very close to the upper bound $\min\left\{ C_{\mathrm{d}}(\rho)/C_{\mathrm{d}}(\sigma),C_{\mathrm{c}}(\rho)/C_{\mathrm{c}}(\sigma)\right\} $
{[}dashed line in Fig.~\ref{fig:bounds}{]}. 
The quality of our bound should also be compared to the lower bound
$C_{\mathrm{d}}(\rho)/C_{\mathrm{c}}(\sigma)$ {[}dotted line in Fig.~\ref{fig:bounds}{]}.
The Figure clearly shows that the two different lower bounds have their advantages for
different values of the parameter $q$: For $q$ close to $1/4$, our new bound is much tighter
than the best previously known bound~\cite{WinterPhysRevLett.116.120404}. If $q$ is below a critical value, the new bound is zero. This corresponds to the region outside the reachable ellipsoid. In addition, the new bound can never exceed one, and thus the results from~\cite{WinterPhysRevLett.116.120404} give a better bound when $\sigma$
has much lower coherence than $\rho$, which corresponds to $q\approx1/2$.

Indeed, we note that for $q=1/4$ the conversion probability $P(\rho\rightarrow\sigma)$
coincides with $C_{\mathrm{c}}(\rho)/C_{\mathrm{c}}(\sigma)$,
and in fact both are equal to $1$. This implies that the asymptotic
conversion rate is given by $R(\rho\rightarrow\sigma)=1$ in this
case. We will generalize this observation in the following Theorem.
\begin{thm}\label{thm:assympUnitRate}
	Assume qubit states $\rho$ and $\sigma$ obey 
	\begin{equation}
	s_{z}^{2}\leq r_{z}^{2}\,\,\,\mathrm{and}\,\,\,s=r.\label{eq:asymptotic}
	\end{equation}
	Then we have $R(\rho\rightarrow\sigma)=1$.
\end{thm}
\begin{proof}
	In the first step of the proof note that $P(\rho\rightarrow\sigma)=1$
	for any two states $\rho$ and $\sigma$ fulfilling Eqs.~(\ref{eq:asymptotic}),
	which follows directly from Eqs.~(3a) and (3b) in~\cite{PhysRevLett.119.140402}.
	This proves that $R(\rho\rightarrow\sigma)\geq1$ in this case. 
	
	In the next step we will show that states fulfilling Eqs.~(\ref{eq:asymptotic})
	have equal coherence cost: 
	\begin{equation}
	C_{\mathrm{c}}(\rho)=C_{\mathrm{c}}(\sigma).\label{eq:Cc}
	\end{equation}
	Since $C_{\mathrm{c}}(\rho)/C_{\mathrm{c}}(\sigma)$ is an upper bound
	on the conversion rate, this will then complete the proof of the Theorem.
	For proving Eq.~(\ref{eq:Cc}), note that $r^2=r_{x}^{2}+r_{y}^{2}=4|\rho_{01}|^{2}$.
	Thus, Eqs.~(\ref{eq:asymptotic}) directly imply the equality $|\rho_{01}|^{2}=|\sigma_{01}|^{2}$.
	Now note that for any single-qubit state $\rho$ the coherence cost
	is a simple function of $|\rho_{01}|^{2}$, see Eq.~(\ref{eq:CcQubit}).
	This completes the proof of Eq.~(\ref{eq:Cc}) and also the proof
	of the Theorem. 
\end{proof}

As we show now, the above Theorem cannot be formulated as an if and only if statement.
From Eq.~(\ref{eq:coherenceFormation}) follows that the coherence cost of a pure state $\psi$ is equal to $S(\Delta[\psi])$ and therefore equal to its distillable coherence (compare Eq.~(\ref{eq:distillableClosed})). Now there exist pure states $\psi$ and mixed states $\rho$ such that 
\begin{align}
0<C_\mathrm{d}(\psi)=C_\mathrm{d}(\rho).
\end{align}
Using Eq.~(\ref{eq:bound-2}), we can conclude
\begin{align}
1=\frac{C_{\mathrm{d}}(\rho)}{C_{\mathrm{d}}(\psi)} =\frac{C_{\mathrm{d}}(\rho)}{C_{\mathrm{c}}(\psi)} \leq R(\rho\rightarrow\psi)\leq \frac{C_{\mathrm{d}}(\rho)}{C_{\mathrm{d}}(\psi)}=1,
\end{align}
which proves $R(\rho\rightarrow\psi)=1$. However, this case is not covered by Thm.~\ref{thm:assympUnitRate}.

\begin{figure}
	\includegraphics[height=4cm]{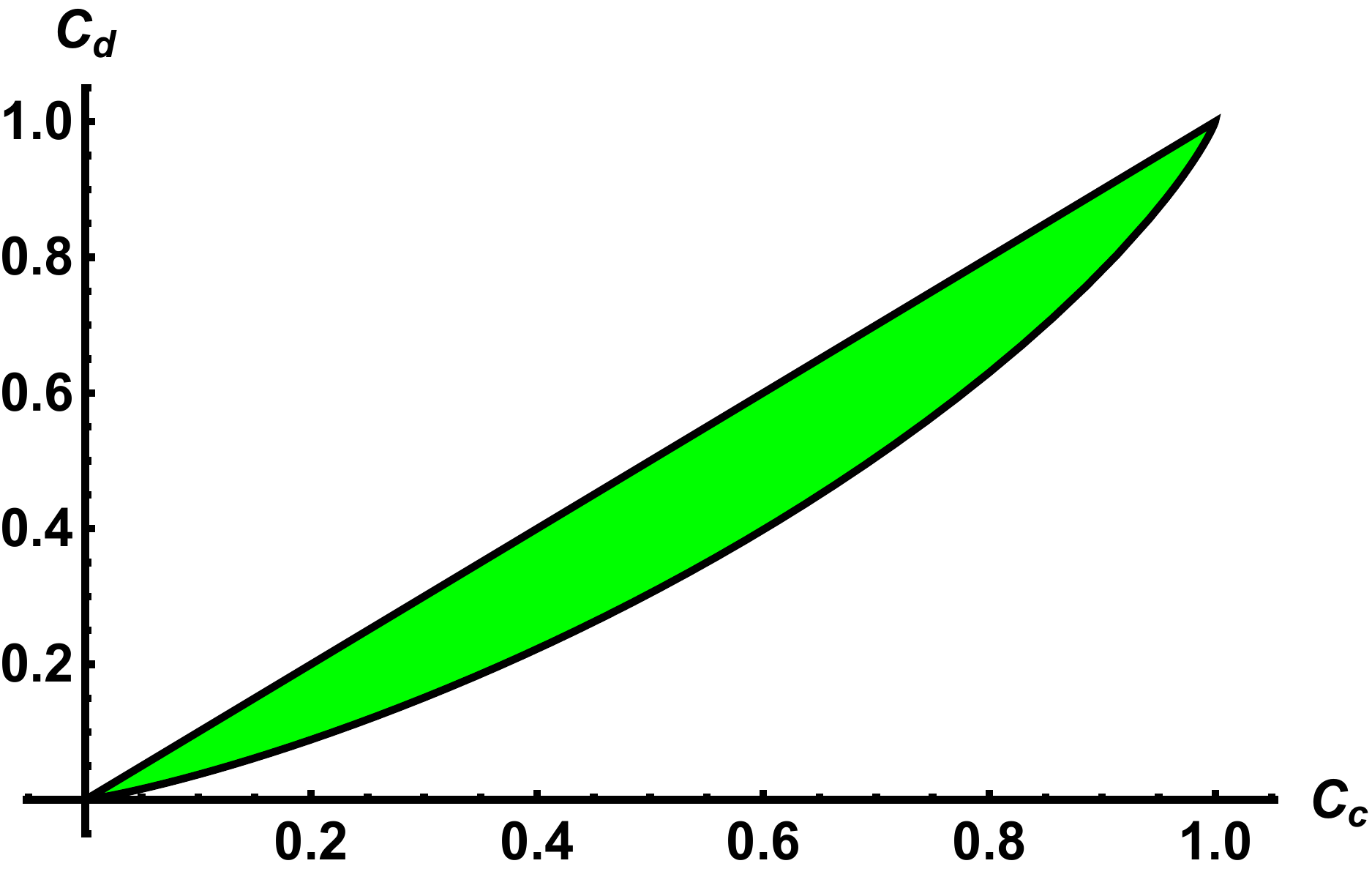}
	
	\caption{\label{fig:irreversibility}Allowed region for distillable coherence
		$C_{\mathrm{d}}$ and coherence cost $C_{\mathrm{c}}$ for single-qubit
		states. The upper curve is given by $C_{\mathrm{d}}(\rho)=C_{\mathrm{c}}(\rho)$,
		which is attained for pure states. The lower curve is obtained from
		the family of states given in Eq.~(\ref{eq:ExampleSigma}), see Prop.~\ref{prop:reversiblityQubits} and its discussion for details.}
\end{figure}
We will now apply the methods we developed for studying
the irreversibility of coherence theory. For any quantum resource
theory, the conversion rate $R$ fulfills the following inequality
for any two nonfree states $\rho$ and $\sigma$: 
\begin{equation} \label{eq:irrevers}
R(\rho\rightarrow\sigma)\times R(\sigma\rightarrow\rho)\leq1.
\end{equation}
The resource theory is called \emph{reversible} if Eq.~(\ref{eq:irrevers})
is an equality for all nonfree states. Otherwise, the resource theory
is called \emph{irreversible}. Examples for reversible resource theories
are the theories of entanglement and coherence, when restricted to
pure states only. However, both theories are not reversible for general
mixed states~\cite{HorodeckiPhysRevLett.80.5239,WinterPhysRevLett.116.120404}.
General properties of reversible resource theories have been investigated
in~\cite{HorodeckiPhysRevLett.89.240403,PhysRevLett.115.070503}.

In the following, we will study the irreversibility of coherence theory
in more detail. In particular, we will investigate which values of
distillable coherence $C_{\mathrm{d}}$ a single-qubit state can attain,
for a fixed amount of coherence cost $C_{\mathrm{c}}$. The most interesting
family of states in this context is given by $\sigma$ in Eq.~(\ref{eq:ExampleSigma}):

\begin{prop}\label{prop:reversiblityQubits}
	Among all single-qubit states, the family of states given in Eq.~(\ref{eq:ExampleSigma}) has the minimal distillable coherence $C_{\mathrm{d}}$
	for a fixed coherence cost $C_{\mathrm{c}}$ and vice versa maximal $C_\mathrm{c}$ for fixed $C_\mathrm{d}$.
\end{prop}
\begin{proof}
	In the first step, we recall that for any single-qubit state $\rho$
	the coherence cost depends only on the absolute value of the offdiagonal
	element $|\rho_{01}|=|\braket{0|\rho}{1}|$, see Eq.~(\ref{eq:CcQubit}). In particular, $C_\mathrm{c}$ is a strictly monotonically increasing function of $|\rho_{01}|$. Moreover, recall that $|\rho_{01}|$ is directly related
	to the Euclidian distance of the state to the incoherent axis in the
	Bloch space: $r_{x}^{2}+r_{y}^{2}=4|\rho_{01}|^{2}$~\footnote{Compare also the proof of Thm.~\ref{thm:assympUnitRate}.}.
	This means that all states with a fixed coherence cost have the same
	distance to the incoherent axis in the Bloch space. 
	
	In the next step, we note that for any single-qubit
	state $\rho$ with Bloch vector $\boldsymbol{r}=(r_x,r_y,r_z)^T$ we can introduce the state $\tilde{\rho}$ having the Bloch coordinates 
	\begin{equation}
	\tilde{r}_x=\sqrt{r_x^2+r_y^2},\,\,\,\,\tilde{r}_y=0,\,\,\,\,\tilde{r}_z=r_z.
	\end{equation}
	The state $\tilde{\rho}$ can be obtained from $\rho$ via an incoherent unitary, and thus both states have the same coherence cost and distillable coherence. In the next step, we introduce the state $\tau$ as follows:
	\begin{equation}
	\tau=\frac{1}{2}\tilde{\rho}+\frac{1}{2}\sigma_{x}\tilde{\rho}\sigma_{x}.
	\end{equation}
	Note that $\tau$ has the same distance to the incoherent axis -- and thus the
	same coherence cost -- as $\rho$ and $\tilde{\rho}$, i.e., 
	\begin{equation}
	C_\mathrm{c} (\tau) = C_\mathrm{c} (\tilde{\rho}) = C_\mathrm{c} (\rho).
	\end{equation}
	Moreover, it is straightforward to
	see that $\tau$ lies on the maximally coherent plane, i.e., the plane spanned by Bloch vectors corresponding to maximally coherent states. By construction, the Bloch vector of $\tau$ also lies in the $x$-$z$ plane, which implies that $\tau$ has the desired form~(\ref{eq:ExampleSigma}).
	
	In the final step, recall that the distillable
	coherence is convex, and thus
	\begin{equation}
	C_{\mathrm{d}}(\tau)\leq\frac{1}{2}C_{\mathrm{d}}(\tilde{\rho})+\frac{1}{2}C_{\mathrm{d}}(\sigma_{x}\tilde{\rho}\sigma_{x})=C_{\mathrm{d}}(\tilde{\rho})=C_{\mathrm{d}}(\rho),
	\end{equation}
	where we used the facts that the Pauli matrix
	$\sigma_{x}$ is an incoherent unitary, and thus preserves~$C_{\mathrm{d}}$, and that $\rho$ and $\tilde{\rho}$ have the same distillable coherence. This completes the proof.
\end{proof}

This result allows us to plot the allowed region of coherence cost
and distillable coherence in Fig.~\ref{fig:irreversibility}. The
upper curve is given by $C_{\mathrm{d}}(\rho)=C_{\mathrm{c}}(\rho)$,
which is attained if $\rho$ is a pure state.
From results in~\cite{PhysRevLett.115.020403,WinterPhysRevLett.116.120404} follows directly, that the same region is attainable for distillable entanglement and entanglement cost when considering maximal correlated two-qubit states.

\section{Experimental Aspects}
In this section, we describe the experimental details, which include: state preparation, implementation of the incoherent Kraus operators and quantum state tomography. 
\subsection{Details on state preparation}
In  module (I) described in the main text, we can prepare three different classes of states. The first class consists of single-qubit states of the form 
\begin{equation}\label{suppsinqubit}
\rho^B=\frac{1}{2}\left(\openone+r_x\sigma_x+r_z\sigma_z\right)
\end{equation}
on Bob's side, where $r_{x,z}$ are real numbers and denote $x, z$ Bloch coordinates. The second class consists of pure two qubit entangled states of the form 
\begin{equation}\label{supptwoqubit}
\ket{\Psi}^{AB}=\sqrt{\mu_0}\ket{0}^A\ket{\beta_0}^B+\sqrt{\mu_1}\ket{1}^A\ket{\beta_1}^B,
\end{equation}
where $\mu_0, \mu_1$ denote eigenvalues of Bob's local states and $\{\ket{\beta_0}, \ket{\beta_1}\}$ the local basis of Bob. The third class are
two qubit Werner states 
\begin{equation}\label{supptwoqubitwerner}
\rho^{AB}_w=q_\mathrm{w}\ketbra{\phi^+}{\phi^+}+\frac{1-q_\mathrm{w}}{4}\openone^{AB},
\end{equation}
where $\ket{\phi^+}$ denotes a maximally entangled state and $q_\mathrm{w}$ is the purity of the Werner state.

In particular, two type-I phase-matched $\beta$-barium borate (BBO) crystals, whose optical axes are normal to each other, are pumped by the continuous laser at $404$~nm, with a power of 80~mW, for the generation of photon pairs with a central wavelength at $\lambda$=808~nm via a spontaneous parametric down-conversion process (SPDC). A half-wave plate (H) working at 404~nm set before the lens and BBO crystals is used to control the polarization of the pump laser. Two polarization-entangled photons 
\begin{equation}
\ket{\Phi(\theta)}=\cos2\theta\ket{HH}+\sin2\theta\ket{VV}
\end{equation}
are generated, and then separately distributed through two single-mode fibers (SMF), where one represents Bob and the other Alice. Two interference filters with a 3~nm full width at half maximum (FWHM) are placed to filter out proper transmission peaks. HWPs at both ends of the SMFs are used to control the polarization of both photons. A quarter-wave plate in Bob's arm is used to compensate the phase for the desired prepared state.

For preparing single-qubit states in Eq.~(\ref{suppsinqubit}), we set the rotation angle of the $404$~nm HWP to $0^\circ$, resulting in a state $\ket{H}^A\ket{H}^B$. By using Alice's photons as trigger, we can experimentally generate a pure incoherent state $\ket{H}^B$. We replace HWP$_1$ with a PBS, a $400\lambda$ QP and another two HWPs (H$_{2,3}$) on Bob's side, for generating single-qubit states $\rho^B$. The rotation angle of H$_2$ is set to $\gamma_1$, rotating the state $\ket{H}$ to another pure state 
\begin{equation}
\cos2\gamma_1\ket{H}+\sin2\gamma_1\ket{V}.
\end{equation}
Then after the birefringent crystal, the pure state is completely dephased, resulting in a incoherent mixed state 
\begin{equation}
\cos^22\gamma_1\ketbra{H}{H}+\sin^22\gamma_1\ketbra{V}{V}.
\end{equation}
The rotation angle of the inserted H$_2$ is set to $\gamma_2$, resulting in the transformation 
\begin{equation}
\begin{aligned}
&\ket{H}\longrightarrow\cos2\gamma_2\ket{H}+\sin2\gamma_2\ket{V},\\
&\ket{V}\longrightarrow\sin2\gamma_2\ket{H}-\cos2\gamma_2\ket{V}.
\end{aligned}
\end{equation}
The final prepared state reads
\begin{equation}
\begin{aligned}
&\rho^B=(\cos^2\gamma_1\cos^2\gamma_2+\sin^2\gamma_1\sin^2\gamma_2)\ketbra{H}{H}\\
&+(\sin^2\gamma_1\cos^2\gamma_2+\cos^2\gamma_1\sin^2\gamma_2)\ketbra{V}{V}\\
&+\frac{1}{2}\cos2\gamma_1\sin2\gamma_2(\ketbra{H}{V}+\ketbra{V}{H}),
\end{aligned}
\end{equation}
with Bloch coordinates
\begin{equation}
\begin{aligned}
&r_x=\cos2\gamma_1\sin2\gamma_2,\\
&r_y=0,\\
&r_z=\cos2\gamma_1\cos2\gamma_2.
\end{aligned}
\end{equation}
Thus we can prepare desired single-qubit states as described in Eq.~(\ref{suppsinqubit}).

For generating two qubit entangled states as given in Eq.~(\ref{supptwoqubit}), we set the rotation angle of the $404$~nm H to $\alpha^\circ$, where $\cos2\alpha=\mu_0$ and $\sin2\alpha=\mu_1$. Then passing through H$_1$ with rotation angle $\beta$, results in $\ket{\Psi(\mu,\beta)}^{AB}$ with desired $\mu$ and $\beta$. Using our experimental setup, the maximally entangled state can be prepared with a fidelity of $0.986$.

For preparing Werner states as in Eq.~(\ref{supptwoqubitwerner}), we make use of an unbalanced Mach - Zehnder interferometer. In particular, two 50/50 beam splitters (BS) are inserted into one branch. In the transmission path, the two-photon state is prepared as the Bell state 
\begin{equation}
\ket{\phi^+}=\frac{1}{\sqrt{2}}(\ket{HH}+\ket{VV})
\end{equation}
when the rotation angle of the 404~nm H is set as $22.5^\circ$. In the reflected path, three $400\lambda$ quartz crystals and a half-wave plate with $22.5^\circ$ are used to dephase the two-photon state into a completely mixed state $\openone^{AB}/4$. The ratio of the two states mixed at the output port of the second BS can be changed by the two adjustable apertures (AA) for the generation of Werner states in Eq.~(\ref{supptwoqubitwerner}) with arbitrary $q_\mathrm{w}$. Out of the state preparation module, the two photons are distributed to Alice and Bob, as shown in Fig.~1 in the main text. Actually, the two BSs are not ideally 50/50, and the transmission rate for $H$ and $V$ polarized photons are not exactly the same, resulting in a decrease of fidelity to $F=0.971$ when we prepare maximally entangled states though we have slightly adjusted the rotation angle of the $404$~nm H. Note that in our experiments, we adopt $0\equiv H$ and $1\equiv V$.

\subsection{Details of the experimental strictly incoherent operations}
The experimental set up for implementing the non-unitary incoherent operations is illustrated in Fig. 1(d) in the main text, which is the combination of BD$_{1,2,3}$ and H$_{4,5,6,7,8,9}$. In our experiments, we focus on strictly incoherent operations of the form
\begin{align}
K_1=\begin{pmatrix}\cos\theta_0 & 0 \\0 & \cos\theta_1 \end{pmatrix},\quad K_2=\begin{pmatrix}0&\sin\theta_1\\ \sin\theta_0&0 \end{pmatrix}.
\end{align}
For experimentally realizing these operators, the angle of H$_{6,8,9}$ is set to $45^\circ$ for applying a bit flip $\sigma_x$ on the polarization, the angle of H$_{4,5}$ are set to $\frac{\theta_0}{2}$ and $\frac{\theta_1+90^\circ}{2}$, and H$_7$ is used for phase compensation, respectively.

Without loss of generality, we suppose initially we have a qubit state
\begin{equation}
\rho_0=\frac{1}{2}\left(\openone+r_x\sigma_x+r_z\sigma_z\right)
\end{equation}
in the basis $\{\ket{H}, \ket{V}\}$. Considering the path degree of freedom, which is a two dimensional system $e_0, e_1$, the overall state can be written as
\begin{equation}
\rho_0\otimes\ketbra{e_0}{e_0},
\end{equation} 
where we assume the initial state is in path $e_0$.
Then BD$_1$ displaces the horizontally polarized component of a photon from the vertical component to a distance of about 6~mm. Accordingly, the quantum state is entangled by a controlled-NOT gate (the polarization encoded qubit is the contolling qubit) acting on the whole state, resulting in
\begin{align}
&\rho_1=\begin{pmatrix}
\frac{1+r_z}{2}&0&0&\frac{r_x}{2}\\
0&0&0&0\\
0&0&0&0\\
\frac{r_x}{2}&0&0&\frac{1-r_z}{2}\\
\end{pmatrix}
\end{align}
Then after H$_{4,5}$, we implement a controlled-rotation operation on the whole state, 
\begin{equation}
\begin{aligned}
&\ketbra{H}{H}\otimes\ketbra{e_0}{e_0}\xrightarrow{\mathrm{H}_{4}}\ketbra{\theta_0^+}{\theta_0^+}\otimes\ketbra{e_0}{e_0},\\
&\ketbra{V}{V}\otimes\ketbra{e_0}{e_0}\xrightarrow{\mathrm{H}_{4}}\ketbra{\theta_0^-}{\theta_0^-}\otimes\ketbra{e_0}{e_0},\\
&\ketbra{H}{H}\otimes\ketbra{e_1}{e_1}\xrightarrow{\mathrm{H}_{5}}\ketbra{\theta_1^+}{\theta_1^+}\otimes\ketbra{e_1}{e_1},\\
&\ketbra{V}{V}\otimes\ketbra{e_1}{e_1}\xrightarrow{\mathrm{H}_{5}}\ketbra{\theta_1^-}{\theta_1^-}\otimes\ketbra{e_1}{e_1}.
\end{aligned}
\end{equation}
where we have
\begin{equation}
\begin{aligned}
&\ket{\theta_0^+}=\cos\theta_0\ket{H}+\sin\theta_0\ket{V},\\
&\ket{\theta_0^-}=\sin\theta_0\ket{H}-\cos\theta_0\ket{V},\\
&\ket{\theta_1^+}=\cos\theta_1\ket{H}+\sin\theta_1\ket{V},\\
&\ket{\theta_1^-}=\sin\theta_1\ket{H}-\cos\theta_1\ket{V}.
\end{aligned}
\end{equation}
Then the overall state after this transformation is given by
\begin{widetext}
	
	\begin{align}
	\rho_2=\begin{pmatrix}
	\frac{1+r_z}{2}\cos^2\theta_0&\frac{r_x}{2}\cos\theta_0\sin\theta_1&\frac{1+r_z}{4}\sin2\theta_0&-\frac{r_x}{2}\cos\theta_0\cos\theta_1\\
	\frac{r_x}{2}\cos\theta_0\sin\theta_1&\frac{1-r_z}{2}\sin^2\theta_1&\frac{r_x}{2}\sin\theta_0\sin\theta_1&-\frac{1-r_z}{4}\sin2\theta_1\\
	\frac{1+r_z}{4}\sin2\theta_0&\frac{r_x}{2}\sin\theta_0\sin\theta_1&\frac{1+r_z}{2}\sin^2\theta_0&-\frac{r_x}{2}\sin\theta_0\cos\theta_1\\
	-\frac{r_x}{2}\cos\theta_0\cos\theta_1&-\frac{1-r_z}{4}\sin2\theta_1&-\frac{r_x}{2}\sin\theta_0\cos\theta_1&\frac{1-r_z}{2}\cos^2\theta_1\\
	\end{pmatrix}
	\end{align}
\end{widetext}
in the basis $\{\ket{H}\otimes\ket{e_0}, \ket{H}\otimes\ket{e_1}, \ket{V}\otimes\ket{e_0}, \ket{V}\otimes\ket{e_1}\}$.

BD$_2$ has the length equal to $\frac{2}{3}$ the length of BD$_{1,3}$, which displace the horizontal component 4~mm away from the vertical component, resulting in the following transformations by extending the system to a higher-dimensional space,
\begin{equation}
\begin{aligned}
&\ket{H}\otimes\ket{e_0}\xrightarrow{\mathrm{BD}_2}\ket{H}\otimes\ket{d_0},\\
&\ket{V}\otimes\ket{e_0}\xrightarrow{\mathrm{BD}_2}\ket{V}\otimes\ket{d_1},\\
&\ket{H}\otimes\ket{e_1}\xrightarrow{\mathrm{BD}_2}\ket{H}\otimes\ket{d_2},\\
&\ket{V}\otimes\ket{e_1}\xrightarrow{\mathrm{BD}_2}\ket{V}\otimes\ket{d_3}.
\end{aligned}
\end{equation}
Then the quantum state after BD$_2$ will be\\

\begin{widetext}
	\begin{equation}
	\begin{aligned}
	\rho_3=\begin{pmatrix}
	\frac{1+r_z}{2}\cos^2\theta_0&0&\frac{r_x}{2}\cos\theta_0\sin\theta_1&0&0&\frac{1+r_z}{4}\sin2\theta_0&0&-\frac{r_x}{2}\cos\theta_0\cos\theta_1\\
	0&0&0&0&0&0&0&0\\
	\frac{r_x}{2}\cos\theta_0\sin\theta_1&0&\frac{1-r_z}{2}\sin^2\theta_1&0&0&\frac{r_x}{2}\sin\theta_0\sin\theta_1&0&-\frac{1-r_z}{4}\sin2\theta_1\\
	0&0&0&0&0&0&0&0\\
	0&0&0&0&0&0&0&0\\
	\frac{1+r_z}{4}\sin2\theta_0&0&\frac{r_x}{2}\sin\theta_0\sin\theta_1&0&0&\frac{1+r_z}{2}\sin^2\theta_0&0&-\frac{r_x}{2}\sin\theta_0\cos\theta_1\\
	0&0&0&0&0&0&0&0\\
	-\frac{r_x}{2}\cos\theta_0\cos\theta_1&0&-\frac{1-r_z}{4}\sin2\theta_1&0&0&-\frac{r_x}{2}\sin\theta_0\cos\theta_1&0&\frac{1-r_z}{2}\cos^2\theta_1\\
	\end{pmatrix}
	\end{aligned}
	\end{equation}
\end{widetext}
in the basis $\{\ket{H}\otimes\ket{d_0}, \ket{H}\otimes\ket{d_1}, \ket{H}\otimes\ket{d_2}, \ket{H}\otimes\ket{d_3}, \ket{V}\otimes\ket{d_0}, \ket{V}\otimes\ket{d_1}, \ket{V}\otimes\ket{d_2}, \ket{V}\otimes\ket{d_3}\}$.
As H$_{6,8}$ are set to $45^\circ$, performing a $\sigma_x$ operation on the polarization state in path $d_{0, 3}$, and H$_7$ is set to $0^\circ$ for applying a $\sigma_z$ operation on the polarization state in path $d_{1, 2}$. \\

\begin{widetext}
	\begin{equation}
	\begin{aligned}
	\rho_4=\begin{pmatrix}
	0&0&0&0&0&0&0&0\\
	0&0&0&0&0&0&0&0\\
	0&0&\frac{1-r_z}{2}\sin^2\theta_1&-\frac{1-r_z}{4}\sin2\theta_1&\frac{r_x}{2}\cos\theta_0\sin\theta_1&-\frac{r_x}{2}\sin\theta_0\sin\theta_1&0&0\\
	0&0&-\frac{1-r_z}{4}\sin2\theta_1&\frac{1-r_z}{2}\cos^2\theta_1&-\frac{r_x}{2}\cos\theta_0\cos\theta_1&-\frac{r_x}{2}\sin\theta_0\cos\theta_1&0&0\\
	0&0&\frac{r_x}{2}\cos\theta_0\sin\theta_1&-\frac{r_x}{2}\cos\theta_0\cos\theta_1&\frac{1+r_z}{2}\cos^2\theta_0&\frac{1+r_z}{4}\sin2\theta_0&0&0\\
	0&0&-\frac{r_x}{2}\sin\theta_0\sin\theta_1&-\frac{r_x}{2}\sin\theta_0\cos\theta_1&\frac{1+r_z}{4}\sin2\theta_0&\frac{1+r_z}{2}\sin^2\theta_0&0&0\\
	0&0&0&0&0&0&0&0\\
	0&0&0&0&0&0&0&0\\
	\end{pmatrix}
	\end{aligned}
	\end{equation}
\end{widetext}

The action of BD$_3$ will coherently combine the paths $d_{0,2}$, $d_{1,3}$, followed by H$_9$ with rotation angle $45^\circ$ which flips the polarization state in path $e_0$. The final quantum state after that two Hs will be\\

\begin{widetext}
	\begin{align}
	\rho_5=\begin{pmatrix}
	\frac{1+r_z}{2}\cos^2\theta_0&\frac{r_x}{2}\cos\theta_0\cos\theta_1&\frac{r_x}{2}\cos\theta_0\sin\theta_1&\frac{1+r_z}{4}\sin2\theta_0\\
	\frac{r_x}{2}\cos\theta_0\cos\theta_1&\frac{1-r_z}{2}\cos^2\theta_1&-\frac{1-r_z}{4}\sin2\theta_1&-\frac{r_x}{2}\sin\theta_0\cos\theta_1\\
	\frac{r_x}{2}\cos\theta_0\sin\theta_1&-\frac{1-r_z}{4}\sin2\theta_1&\frac{1-r_z}{2}\sin^2\theta_1&\frac{r_x}{2}\sin\theta_0\sin\theta_1\\
	\frac{1+r_z}{4}\sin2\theta_0&-\frac{r_x}{2}\sin\theta_0\cos\theta_1&\frac{r_x}{2}\sin\theta_0\sin\theta_1&\frac{1+r_z}{2}\sin^2\theta_0\\
	\end{pmatrix}
	\end{align}
\end{widetext}

Note that we can tilt BD$_3$ for phase compensation, removing the negative signs in the interference terms between the horizontal and vertical component in each arm.
When tracing over the path degree, we obtain\\
\begin{widetext}
	\begin{align}
	\rho_f=\begin{pmatrix}
	\frac{1+r_z}{2}\cos^2\theta_0+\frac{1-r_z}{2}\sin^2\theta_1&\frac{r_x}{2}\cos\theta_0\cos\theta_1+\frac{r_x}{2}\sin\theta_0\sin\theta_1\\
	\frac{r_x}{2}\cos\theta_0\cos\theta_1+\frac{r_x}{2}\sin\theta_0\sin\theta_1&\frac{1+r_z}{2}\sin^2\theta_0+\frac{1-r_z}{2}\cos^2\theta_1\\
	\end{pmatrix}
	\end{align}
\end{widetext}
which is exactly the state after implementing $K_1$ and $K_2$.
Combination of two BDs (BD$_1$ and BD$_3$) results in a natural robust interference, and the phase between $d_0$, $d_2$ and $d_1$, $d_3$ can be removed either by adjusting the position of BD$_{1,3}$ or tilting the Hs in each arm.

Note that in our experiments, if we obtain a final state with Bloch coordinates $\tilde{r}_x>0$ and $\tilde{r}_z>0$, then we can also obtain the state $(-\tilde{r}_x, 0, \tilde{r}_z)$, $(\tilde{r}_x, 0, -\tilde{r}_z)$, and $(-\tilde{r}_x, 0, -\tilde{r}_z)$ via a $\sigma_z$ operation, a $\sigma_x$ operation and the combination of a $\sigma_z$ operation and a $\sigma_x$ operation. In Fig. 1(d), H$_{10}$ is set to $0^\circ$ for implementing a $\sigma_z$ operation, and H$_{11}$ is set to $45^\circ$ for implementing a $\sigma_x$ operation. Namely, we only need to focus on the states with Bloch coordinates $\tilde{r}_x>0$ and $\tilde{r}_z>0$, and complete the whole boundary via simple incoherent operations.

\subsection{State tomography and data collection}
After the implementation of the incoherent operations, we perform quantum state tomography. In particular, when we conduct the experiment without assistance, the single-qubit state after the incoherent operation can be directly identified via the combination of two Hs, two Qs and two PBSs in Module (III). For deterministic state conversion, we directly read the total coincident counts from the two SPDs; and for stochastic state conversion, we discard the counts from $K_2$. For experimentally determining the conversion probability in the case of stochastic conversion, we also collect data in an orthogonal basis. The probability for stochastic conversion can then be evaluated as 
\begin{equation}
P_1=\frac{N_1}{N_{\textrm{total}}}, 
\end{equation}
where $N_1$ denotes the total coincident counts from $K_1$, and $N_{\textrm{total}}$ denotes the total coincident counts from $K_1$ and $K_2$, in basis $\{\ket{H}, \ket{V}\}$.

When we conduct the experiments with assistance, Alice can perform arbitrary local projective measurements on her photons, and broadcast the measurement outcomes to Bob. Specifically, Alice chooses the optimal measurement, which helps Bob to get maximal coherence. When Bob gets the information from Alice, which is either $0$ or $1$, he can then implement the aforementioned incoherent operations, obtaining the final target states.

For data collections, we used single-mode fibers on Bob's arm and multi-mode fibers on Alice's arm for directing photons from space to detectors. The use of multi-mode fibers can increase and stabilize the collection effeciency of Alice's photons.
On the other hand, the use of single-mode fibers on Bob's side is preferable for cleaning up the high order optical mode, resulting in best interference between the light beams which are displaced by the BDs. The power of the 404~nm continues laser is set to about 80~mW, and the coincidence window is set at 4~ns, resulting in 2000 coincident events in a second. When adding white noise on Alice's arm, the coicident counts decrease to around 25\% when compared to the case without noise.

\end{document}